\documentclass[11pt,reqno]{amsart}
\usepackage{amsthm,amssymb,amsfonts,dsfont}
\usepackage[left=25mm,top=25mm,bottom=25mm,right=25mm]{geometry}
\usepackage[english]{babel} 
\usepackage[utf8]{inputenc} 
\usepackage{cmap}           
\usepackage{hyperxmp}       
\usepackage{graphicx}
\usepackage{booktabs}
\usepackage{diagbox}
\usepackage[pdfdisplaydoctitle = true,
      colorlinks = true,
      urlcolor = blue,
      citecolor = blue,
      linkcolor = blue,
      pdfstartview = FitH,
      pdfpagemode = UseNone,
      bookmarksnumbered = true,
      unicode=true]{hyperref} 
\usepackage{bm}
\usepackage[binary-units=true]{siunitx}
\usepackage{subcaption}
\usepackage{xcolor}
\usepackage{tikz}
\usetikzlibrary{shapes} 
\usepackage{enumitem}
\usepackage[pagewise]{lineno}

\graphicspath{{figs/}}
\newtheorem{thm}{Theorem}[section]

\newtheorem{lem}[thm]{Lemma}

\theoremstyle{definition}
\newtheorem{defn}[thm]{Definition}
\theoremstyle{remark}
\newtheorem{rem}[thm]{Remark}
\newtheorem{exa}[thm]{Example}

\numberwithin{equation}{section}

\newcommand{\norm}[1]{\left\Vert#1\right\Vert}
\newcommand{\abs}[1]{\left\vert#1\right\vert}
\newcommand{\set}[1]{\left\{#1\right\}}

\newcommand{\RR}{\mathbb{R}}                            
\newcommand{\NN}{\mathbb{N}}                            

\newcommand{\HH}{\mathbb{H}}                            
\newcommand{\Sym}{\mathbb{S}}                           

\newcommand{\OO}{\mathrm{O}}
\newcommand{\SO}{\mathrm{SO}}
\newcommand{\octa}{\mathbb{O}}

\renewcommand{\vec}{\pmb}
\newcommand{\hh}{\pmb{h}}
\newcommand{\mm}{\pmb{m}}

\newcommand{\vv}{\pmb{v}}
\newcommand{\ee}{\pmb{e}}

\newcommand{\ba}{\mathbf{a}}
\newcommand{\be}{\mathbf{e}}

\newcommand{\bgamma}{\pmb{\gamma}}
\newcommand{\bsigma}{\pmb{\sigma}}
\newcommand{\bepsilon}{\pmb{\epsilon}}
\newcommand{\id}{\bm{1}}

\newcommand{\bC}{\mathbf{C}} 
\newcommand{\bI}{\mathbf{I}} 
\newcommand{\bJ}{\mathbf{J}} 
\newcommand{\bS}{\mathbf{S}} 
\newcommand{\bP}{\mathbf{P}} 

\DeclareMathOperator{\tr}{tr}
\DeclareMathOperator{\rank}{rank}

\tikzstyle{rect}=[draw=black, fill=white, line width=1.5pt, rectangle,
 rounded corners, inner sep=10pt, inner ysep=5pt]
\tikzstyle{fancytitle}=[fill=black, text=white]

\tikzstyle{diam}=[draw=black, fill=white,rotate=0, line width=1.5pt, diamond, inner sep=10pt, inner ysep=5pt]
\tikzstyle{fancytitle}=[fill=black, text=white]

\tikzstyle{elli}=[draw=black, fill=white,rotate=0, line width=1.5pt, ellipse, inner sep=10pt, inner ysep=5pt]
\tikzstyle{fancytitle}=[fill=black, text=white]

\hypersetup{
  pdfauthor = {J. Taurines, M. Olive, R. Desmorat, O. Hubert, B. Kolev}, 
  pdftitle = {Integrity bases for cubic nonlinear magnetostriction}, 
  pdfsubject = {}, 
  pdfkeywords = {Invariants, Strong coupling, Magnetostriction}, 
  pdflang = en, 
  }
\begin{document}

\title{Integrity bases for cubic nonlinear magnetostriction}

\author{J. Taurines}
\address[Julien Taurines]{Université Paris-Saclay, ENS Paris-Saclay, CNRS,  LMT - Laboratoire de Mécanique et Technologie, 91190, Gif-sur-Yvette, France}
\email{julien.taurines@ens-paris-saclay.fr}

\author{M. Olive}
\address[Marc Olive]{Université Paris-Saclay, ENS Paris-Saclay, CNRS,  LMT - Laboratoire de Mécanique et Technologie, 91190, Gif-sur-Yvette, France}
\email{marc.olive@math.cnrs.fr}

\author{R. Desmorat}
\address[Rodrigue Desmorat]{Université Paris-Saclay, ENS Paris-Saclay, CNRS,  LMT - Laboratoire de Mécanique et Technologie, 91190, Gif-sur-Yvette, France}
\email{rodrigue.desmorat@ens-paris-saclay.fr}

\author{O. Hubert}
\address[Olivier Hubert]{Université Paris-Saclay, ENS Paris-Saclay, CNRS,  LMT - Laboratoire de Mécanique et Technologie, 91190, Gif-sur-Yvette, France}
\email{olivier.hubert@ens-paris-saclay.fr}

\author{B. Kolev}
\address[Boris Kolev]{Université Paris-Saclay, ENS Paris-Saclay, CNRS,  LMT - Laboratoire de Mécanique et Technologie, 91190, Gif-sur-Yvette, France}
\email{boris.kolev@math.cnrs.fr}


\date{December 10, 2020}%

\begin{abstract}
  A so-called smart material is a material that is the seat of one or more multiphysical coupling. One of the key points in the development of the constitutive laws of these materials, either at the local or at the global scale, is to formulate a free energy density (or enthalpy) from vectors, tensors, at a given order and for a class of given symmetry, depending on the symmetry classes of the crystal constituting the material or the symmetry of the representative volume element. This article takes as a support of study the stress and magnetization couple ($\bsigma$, $\mm$) involved in the phenomena of magnetoelastic coupling in a cubic symmetry medium. Several studies indeed show a non-monotonic sensitivity of the magnetic susceptibility and magnetostriction of certain soft magnetic materials under stress. Modeling such a phenomenon requires the introduction of a second order stress term in the Gibbs free energy density. A polynomial formulation in the two variables stress and magnetization is preferred over a tensorial formulation. For a given material symmetry class, this allows to express more easily the free energy density at any bi-degree in $\bsigma$ and $\mm$ (\textit{i.e.} at any constitutive tensors order for the so-called tensorial formulation). A rigorous and systematic method is essential to obtain the high-degree magneto-mechanical coupling terms and to build a free energy density function at any order which is invariant by the action of the cubic (octahedral) group. For that aim, theoretical and computer tools in Invariant Theory, that allow for the mathematical description of cubic nonlinear magneto-elasticity, are introduced. Minimal integrity bases of the invariant algebra for the pair $(\vec m, \bsigma)$, under the proper (orientation-preserving) and the full cubic groups, are then proposed. The minimal integrity basis for the proper cubic group is constituted of 60 invariants, while the minimal integrity basis for the full cubic group (the one of interest for magneto-elasticity) is made up of 30 invariants. These invariants are formulated in a (coordinate free) intrinsic manner, using a generalized cross product to write some of them. The counting of independent invariants of a given multi-degree in $(\mm, \bsigma)$ is performed. It is shown accordingly that it is possible to list without error all the material parameters useful for the description of the coupled magnetoelastic behavior from the integrity basis. The technique is applied to derive general expressions $\Psi^\star(\bsigma, \mm)$ of the free energy density at the magnetic domains scale exhibiting cubic symmetry. The classic results for an isotropic medium are recovered.
\end{abstract}

\maketitle

\section*{Introduction}
\label{sec:intro}

A so-called \textit{smart} material is a material which has one or more properties making them adaptive and/or evolutive, and which can be modified in a controlled manner by mechanical stresses, temperature, humidity, pH, electric or magnetic field. This includes magnetostrictive materials~\cite{Dap2004}, classical or magnetic shape memory alloys~\cite{Lex2013}, piezoelectric materials~\cite{Uch2017}, multi-ferroic media~\cite{CDB2008}, etc. Some of them are at the functioning basis of sensors or actuators, others are used for energy production or harvesting, and their properties are often intrinsic, resulting from a so-called \textit{multiphysic} coupling expressed through a constitutive law~\cite{DHW2014}.

The study of electro-magneto-mechanical coupling phenomena is thus of growing interest in the recent years, in particular given the development of computing capacities and the optimization perspectives that their taking into account allows to glimpse. Indeed, the response of a magnetic or dielectric material, if it is at first order function of the magnetic field or the electric field, also depends very strongly on the applied stresses~\cite{DHBB2008,DHW2014}. These coupling effects can sometimes constitute an issue (they are for example responsible for part of the vibrations of rotating machines or the noise emitted by electrical transformers~\cite{LHM2017}), they can also be exploited to serve as a basis for the design of innovative devices for energy conversion, health control of structures or actuators~\cite{schwartz2002}, and even for the design of new materials (often composites)~\cite{NBDM2011}.

The construction of constitutive laws often calls for the writing of a free energy (or enthalpy) density whose derivative with respect to state variables produces the associated variables~\cite{LCBD2009}. The electro-magneto-mechanical coupled terms in the expression of free energy thus most often involve isotropic invariant forms combining vectors or pseudo-vectors\footnote{The magnetic field, magnetic induction or magnetization are pseudo-vectors, which means that an improper orthogonal transformation $g$ ($\det g=-1$) acts on them as $(\det g) g\,\vv=-g\,\vv$ (while $g$ acts on a vector as $g\,\vv$). Distinguishing between vectors and pseudo-vectors is important when considering material symmetries, because a symmetry plane transforms, for instance, a direct basis into an indirect one.} (electric field, magnetic field) and second-order symmetric tensors (strain tensor for example). Furthermore, ferromagnetic or ferroelectric materials are often polycrystals, each crystal exhibiting a certain degree of symmetry and the assembly of which can lead to more or less strong symmetries than the symmetry observed at the local scale. On the other hand, a decomposition of the microstructure in domains (magnetic or ferroelectric) can be considered~\cite{HS2008,DHW2014}. Given the complexity of the materials, macroscopic approaches have long been key for the numerical simulation of the coupled behavior of these materials. Multiscale approaches, using specific localization and homogenization rules, have developed thanks to the general reduction in computation time, and because they often allow the observed experimental behaviors to be more accurately represented~\cite{DHBB2008,DHW2014,Hub2019} without requiring full field approach (micromagnetism, phase field) that still remain highly time-consuming. In these approaches, the behavior is described at a scale where most of the fields can be considered as homogeneous but where the material cannot be considered as isotropic.

One of the key points in the development of constitutive laws at the local or global scale is therefore to formulate a free energy (or enthalpy) density from vectors, tensors, at a given order and for a given symmetry class, depending on the symmetry classes of the crystal constituting the material or the symmetry of the representative volume element, of course avoiding to forget terms in the operation. This article deals with the way in which it is possible to list without error all the invariants and thus all the material parameters useful for the description of the coupled behavior, involving one second-order symmetric tensor (e.g. the stress tensor $\bsigma$), and a vector or a pseudo-vector (e.g. the magnetization $\vec m$).

We will therefore need a minimal integrity basis of the invariant algebra for pair $(\vec m, \bsigma)$ under a certain group $G$ (that can be shown to be a subgroup of $\OO(3)$ even if magnetic point groups are in fact involved \cite{LL1960,Bir1964,SC1984,KE1990,WG2004}, see section \ref{sec:mathematical-modeling}). Many works concern the determination of such bases for matrices and (pseudo-)vectors. As we will focus on cubic symmetry we will rely on  the work of Smith, Smith and Rivlin \cite{SSR1963} for the crystal symmetry classes. Making use of modern computational means, our work will be to check (and correct) these authors results for the proper cubic (octahedral) group $G=\octa^+\subset \SO(3)$ and for the full cubic group $G=\octa$. One achievement of the present paper is an intrinsic (frame independent) writing of the corresponding cubic integrity bases. To the best of our knowledge this point was not addressed in the literature. We give an exact evaluation of the number of terms / of material parameters to be considered in a cubic free energy (or enthalpy) density for non-linear magneto-elasticity.

\subsection*{Outline}

We present in the next section a recent example of introduction of a second order stress term in the Gibbs free energy density of a magnetic domain within a cubic crystal, and conclude in the limitations of such a constitutive tensor based approach. This example highlights the requirement for a rigorous and systematic method to obtain the high degree magneto-mechanical coupling terms. We then introduce theoretical and computational tools (\autoref{sec:mathematical-modeling}) that allow for the mathematical description of cubic nonlinear magneto-elasticity. The \autoref{sec:integrity-basis} is especially dedicated to the construction of a minimal integrity basis for the cubic groups $\octa^+$ and $\octa$, allowing for the closed form calculation of any magnetization/stress degree terms in Gibbs free energy. We turn back to cubic magneto-elasticity applications in~\autoref{sec:domain-scale-coupling}, dealing with modeling at the domain scale. Finally, three appendices \autoref{sec:proofs}, \autoref{sec:algo} and \autoref{sec:degree-bounds} have been added to gather the technical details required to prove the two main theorems \ref{thm:main-octa} and \ref{thm:main-octa-plus}.

\subsection*{Preamble}

In this paper all the energy quantities are energy densities (J/m$^3$). Considering thermo-magneto-mechanics, the arguments of the internal energy density $u$ are the following thermodynamics variables: the entropy density $s$, the magnetic induction $\vec b$ and the (small) strain $\bepsilon$. The introduction of \emph{(Helmholtz) free energy density} $\Psi=u-Ts$, by standard Legendre transform, allows for the energy density to be defined as a function of the absolute temperature $T$ instead of the entropy density. The \textit{magnetic free enthalpy} $k=\Psi-\vec h \cdot \vec b$ (Legendre transform on magnetic dual variables $(\vec h, \vec b)$) allows for the energy density to be defined as a function of the magnetic field $\vec h$ instead of the magnetic induction. The Gibbs free energy (or free enthalpy) density $g=k-\bsigma:\bepsilon$ (Legendre transform on mechanical dual variables $(\bsigma, \bepsilon)$) allows for the energy to be defined as a function of the stress $\bsigma$ instead of the strain. On the other hand, due to the decomposition $\vec b=\mu_0(\vec h + \vec m)$ (where $\mu_0$ is the vacuum magnetic permeability), a common simplification is to consider only material contributions to the energy variation and not magnetic field contributions. The \textit{free energy density} $\Psi^\star(\bsigma, \mm)$ will be defined as a function of magnetization and stress, \textit{i.e.} considered at a given (reference) temperature.

\section{Magneto-elastic coupling modeling: state of the art and objectives}
\label{sec:mechanical-modeling}

The magneto-elastic coupling results in the existence of a deformation of magnetic origin: the magnetostriction strain
\begin{equation*}
  \bepsilon^{\mu}=\bepsilon-\bepsilon^{e},
\end{equation*}
where $\bepsilon$ and $\bepsilon^{e}$ rely to total and elastic strains respectively, and a combined effect of the mechanical stress $\bsigma$ on the magnetic behavior\footnote{A different approach consists in considering the magnetostriction strain as the effect of a combination of volume forces and torques following the work of Eringen and Maugin in the Seventies~\cite{Mau1990}. In this description magnetostriction strain is an elastic strain function of magnetic and electric quantities at the very local (atomic) scale. This model is not retained for the description of the magnetostriction strain in this paper. It is defined as a free strain dependent on material constants as dilatation coefficients for a thermal strain.}.

The modeling of these phenomena has been the subject of numerous works at different scales~\cite{Boz1951,Cul1972,DTdLac1993,GHBB1998,Arm2002,DKT2012,DHBB2008}. The multiscale approaches have known recent developments~\cite{DRH2014,Dap2004,Hub2019}. For some of them, the magnetic domains volume fraction is calculated from the Gibbs free energy density in each domain by a stochastic approach. A magnetic domain (Weiss domain) can be defined as a material volume inside a ferromagnetic material (they are however present in ferrimagnetic and antiferromagnetic materials) where the magnetization is uniform in direction and magnitude, equal to the saturation magnetization considered as a material constant at the room temperature \cite{HS2008}. At the demagnetized state, the matter is divided in numerous magnetic domains whose organization leads to an null average magnetization. Magnetic domains inside a grain are separated by domain walls (Bloch walls) where the magnetization rotates from one direction to another usually associated with the crystallographic axes. Their thickness is usually negligible comparing to the other dimensions of the domain. \autoref{fig:fe27codomain} illustrates the domains organization at the surface of a Fe-27\%Co polycrystalline alloy. This picture has been obtained using a Kerr effect set-up \cite{Sav2018}. Each domain defines its own magnetization but also its own deformation by spin-orbit coupling mechanism at the atomic scale \cite{DTdLac1993}. It can be considered as a uniform free deformation over the domain.

\begin{figure}[ht]
  \centering
  \includegraphics[width=0.65\textwidth]{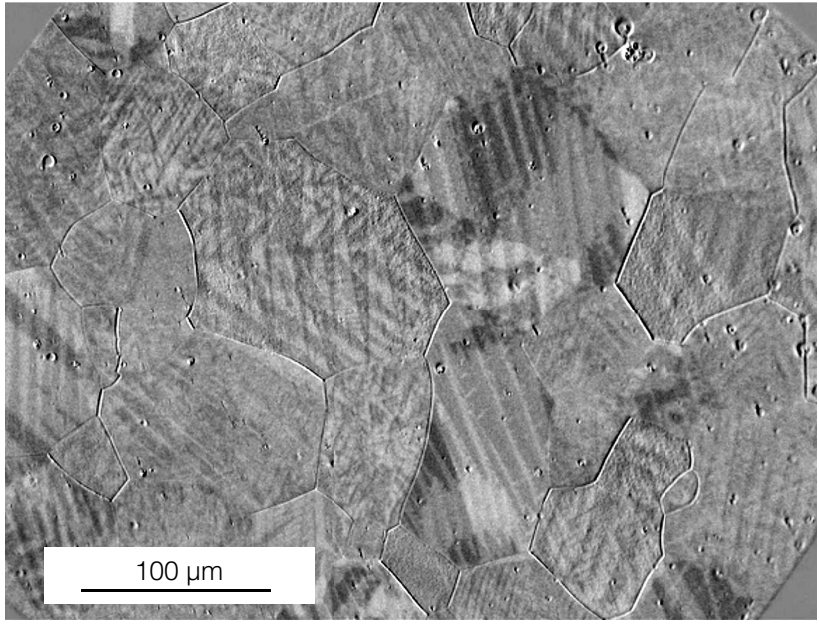}
  \caption{Illustration of the magnetic domains structure at the surface of a Fe-27\%Co polycrystal \cite{Sav2018} - Kerr microscopy.}
  \label{fig:fe27codomain}
\end{figure}

We here focus on existing modeling \emph{at the magnetic domain scale} in which the local magnetization $\mm$ ---defining the macroscopic magnetization $\vec M=\langle \mm \rangle$ by spatial averaging--- has a constant norm, equal to the saturation magnetization $m_{s}$. The equality
\begin{equation}\label{eq:mi}
  \mm= m_{s} \, \bgamma, \qquad \| \bgamma \|=1,
\end{equation}
defines then the magnetization direction vector $\bgamma$.

In a first approximation, the magnetostriction is often considered as a stress independent free strain $\bepsilon^{\mu}$, quadratic in magnetization
\cite{Cul1972,DTdLac1993} (linear terms in $\mm$ are forbidden by the invariance of the underlying microscopic systems under time reversal \cite{LL1960,Bir1964,Eri1976}).
The corresponding magneto-elastic free energy density $ \Psi^{\star\mu\sigma}(\mm,\bsigma)$, defined at the magnetic domain scale,  is:
\begin{equation}\label{gibbslin}
  \Psi^{\star\mu\sigma}(\mm,\bsigma)=-\bsigma:{\pmb{\mathcal{E}}}:(\mm \otimes \mm)=-\bsigma:\bepsilon^{\mu}
\end{equation}
The magnetostriction strain is obtained as the derivative, with respect to the stress, of the magneto-elastic density:
\begin{equation}
  \bepsilon^{\mu}=-\frac{\partial \Psi^{\star\mu\sigma}}{\partial \bsigma}=\pmb{\mathcal{E}}:(\mm \otimes \mm).
\end{equation}
$m_s^2\,\pmb{\mathcal{E}}$ is the fourth-order magnetostriction tensor ($\mathcal{E}_{ijkl}=\mathcal{E}_{jikl}=\mathcal{E}_{ijlk}$), function of three independent material parameters in the cubic symmetry case. This number reduces to two constants (the so-called magnetostriction constants $\lambda_{100}$ and $\lambda_{111}$~\cite{Boz1951}) when further incompressibility condition $\tr\bepsilon^{\mu}=0$ is considered.

It is moreover observed that for some different iron-based ferromagnetic materials (such as iron-silicon, iron-cobalt, steels) the mechanical stress has a non-monotonic effect on the magnetic behavior: such a non-monotony is characterized by a sudden decrease of the initial (macroscopic) magnetic susceptibility with increasing stress. This point is illustrated in \autoref{fe27conew} for Fe-27\%Co polycrystal.

\begin{figure}[ht]
  \centering
  \begin{subfigure}{0.49\textwidth}
    \includegraphics[width=1\textwidth]{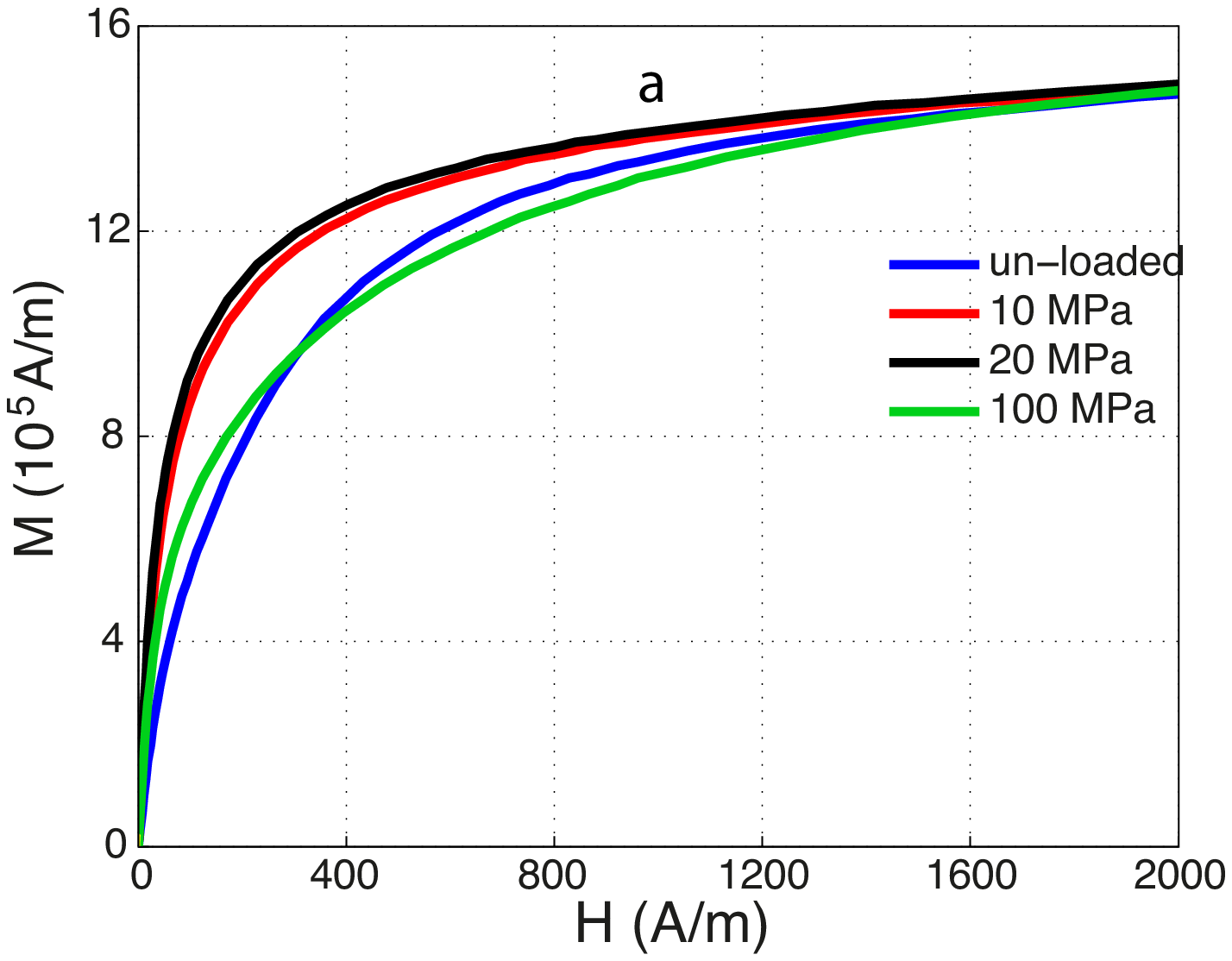}

  \end{subfigure}
  \begin{subfigure}{0.49\textwidth}
    \includegraphics[width=1\textwidth]{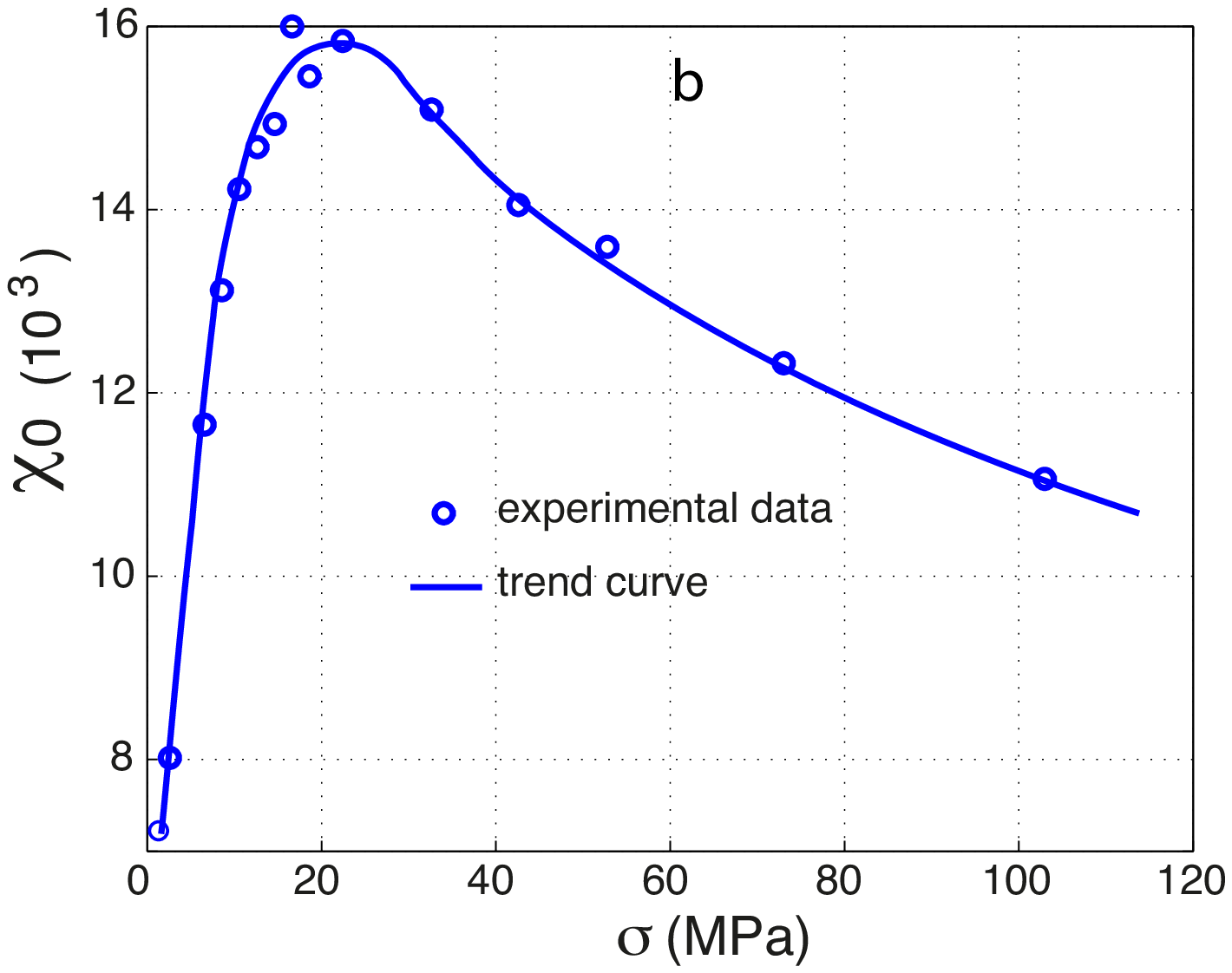}
  \end{subfigure}
  \caption{(a) Anhysteretic magnetic behavior (magnetization $M$ as function of applied magnetic field $H$ in the direction of applied field) and (b) initial anhysteretic susceptibility $\chi_{0}$ ($=\partial M/\partial H\lvert_{H=0}$) of a Fe-27$\%$Co polycrystal\protect\footnotemark \, subjected to increasing uniaxial stress levels applied along the magnetic field direction \cite{savary}.}
  \label{fe27conew}
\end{figure}
\footnotetext{The microstructure of this material is presented in \autoref{fig:fe27codomain}.}

It has been shown in \cite{Hub2019} that a magneto-elastic term as defined by equation \eqref{gibbslin} is unable to describe such a non-monotony effect. This observation suggests the presence of magneto-elastic coupling terms of degree larger than 1 in $\bsigma$, as Mason~\cite{Mas1951} did suppose.

The definition (\ref{gibbslin}) of the magneto-elastic free energy density corresponds indeed to a first order Taylor expansion in $\bsigma$. A second order development in stress has been considered in~\cite{Hub2019}, representing the so-called \textit {morphic effect}~\cite{DTdLac1993}. Its practical implementation involved the identification of the components of a sixth-order tensor $\mathbb{E}$ ($m_{s}^{2}\,\mathbb{E}$ being the \textit{morphic tensor}), replacing expression \eqref{gibbslin} by
\begin{equation}\label{Gibbsnonlin}
  \Psi^{\star\mu\sigma}(\mm,\bsigma)=
  -\bsigma:\,\pmb{\mathcal{E}}:(\mm \otimes \mm)
  -\frac{1}{2}\,\bsigma:(\bsigma:\,\mathbb{E}:(\mm \otimes \mm))
\end{equation}
with then
\begin{equation*}
  \bepsilon^{\mu} = -\frac{\partial \Psi^{\star\mu \sigma}}{\partial \bsigma}=\pmb{\mathcal{E}}:(\mm \otimes \mm)
  +
  \,\bsigma:\mathbb{E}:(\mm \otimes \mm)
\end{equation*}

A sixth-order constitutive tensor in the Euclidian space introduces $3^{6}$ (=729) material parameters. This number can be strongly reduced using the index symmetries of the stress and orientation tensors and grand symmetries (namely $\mathbb{E}_{ijklmn}=\mathbb{E}_{jiklmn}=\mathbb{E}_{ijlkmn}=\mathbb{E}_{ijklnm}$ and $\mathbb{E}_{ijklmn}=\mathbb{E}_{klijmn}$).

A further reduction can be achieved by taking into account the crystal symmetry of the material (using Curie principle \cite{Cur1894}). This work has been performed by Mason~\cite{Mas1951} for cubic symmetry and by Kraus~\cite{Kra1988} for isotropy. The methods employed for these reductions lead however to some issues, not strictly limited to magneto-elasticity (nonlinear piezo-electricity is also concerned):
\begin{itemize}
  \item It is quite difficult to ascertain the exact number of independent material constants (\emph{i.e.} independent components of the morphic tensor);

  \item When several physics are involved, as for the magneto-elastic free energy density $\Psi^{\star\mu \sigma}(\mm, \bsigma)$, it cannot be guaranteed that some (joined) invariants are not forgotten when the degrees of the truncated Taylor expansions are increased;

  \item Related to previous point, the determination of closed form expressions for free energy densities involving high order tensors is often tricky: using constitutive tensors of order higher than 6 becomes almost impossible by this technique;

  \item The same remark applies if a material exhibits lower symmetry (such as orthotropy or as monoclinic symmetry).
\end{itemize}

In the following, we focus our efforts on the determination of the relevant invariants for cubic symmetry of a pair $(\mm,\bsigma)$ with $\mm$ either a vector or a pseudo-vector and $\bsigma$ a second-order symmetric tensor. The general mathematical framework allowing for the correct writing of constitutive models involving material symmetries is known as \emph{Invariant Theory}~\cite{Stu1993,DK2015} and corresponds to the study of tensors subjected to group actions. If this theory goes back to the nineteenth century, the exponential complexity of the calculations (illustrated in~\autoref{sec:algo}) explains that applications are limited. This barrier has now been lifted thanks to current computing power (see for instance~\cite{Bed2007,BP2010,Oli2017,OKA2017,LO2017} for recent computations in this field).

\section{Mathematical modeling of cubic constitutive laws involving a vector and a second-order tensor}
\label{sec:mathematical-modeling}

\subsection{Cubic symmetry, decomposition and projectors}\label{sec:mathmodeling-sub1}

The symmetry group of a material is the subgroup $G$ of isometries leaving its crystal lattice invariant. According to the Curie principle~\cite{Cur1894}, constitutive laws inherit material symmetries. In particular, when these laws can be described by an energy density, this function is invariant by the material symmetry group $G$. \emph{Isotropy} as a possible symmetry of the model corresponds to invariance of the energy density by the full group of orthogonal transformations $\OO(3)$. \emph{Hemitropy} (see Zheng and Boehler in~\cite{ZB1994}), as another possibility, corresponds to invariance under the subgroup $\SO(3)$ of $\OO(3)$, of orientation-preserving orthogonal transformations (\textit{i.e.} rotations).

In this paper, the free energy density $\Psi^\star$ is a function of magnetization $\mm$ and stress $\bsigma$ and we focus on cubic symmetry which is described by the octahedral (or cubic) group $\octa$, defined as
\begin{equation*}
  \octa = \set{g\in \OO(3);\; g \ee_{i} = \pm \ee_{j}},
\end{equation*}
where $(\ee_{i})$ is the canonical orthonormal basis of $\RR^{3}$. The cubic group $\octa$ is composed of $48$ elements: $24$ rotations, which form a sub-group of $\octa$ denoted by $\octa^{+}$, and $24$ orientation-reversing isometries. A cubic invariant free energy density is thus a function $\Psi^\star$ of $\mm$ and $\bsigma$ so that
\begin{equation}\label{eq:energy}
  \Psi^\star(g \star \mm, g \star \bsigma)=\Psi^\star(\mm,\bsigma), \qquad \forall g\in \octa ,
\end{equation}
where the action of $g$ on magnetization and stress writes
\begin{equation}\label{eq:pseudo-vector-action}
  g\star \mm = (\det g)\, g \mm, \quad \text{and} \quad g \star \bsigma = g \bsigma g^{t}.
\end{equation}
when $\mm$ is a pseudo-vector, and
\begin{equation}\label{eq:vector-action}
  g\star \mm = g \mm, \quad \text{and} \quad g \star \bsigma = g \bsigma g^{t}.
\end{equation}
when $\mm$ is a vector.

The pair $(\mm,\bsigma)$ spans the vector space $V=\RR^{3} \oplus \Sym^{2}(\RR^{3})$, where $\Sym^{2}(\RR^{3})$ is the space of symmetric second-order tensors. This vector space $V$ can be furthermore split into \emph{irreducible components}, \textit{i.e.} stable subspaces that contains no other stable subspaces other than themselves or the trivial subspace $\set{0}$. Concerning the action of the cubic group $\octa$ on $V$, the stable subspace $\RR^{3}$ (associated to $\mm$) is itself irreducible but the stable subspace $\Sym^{2}(\RR^{3})$ (associated to $\bsigma$) is made up of three irreducible subspaces (see~\cite{Mug1972,FK1977}). The first one is one-dimensional and spanned by $(\tr\bsigma)\id$. The second one is two-dimensional and consists of diagonal deviators $\bsigma^{d}$. Finally, the third one is three-dimensional and consists of symmetric deviators with vanishing diagonal $\bsigma^{\overline{d}}$. These tensors are defined respectively in the canonical basis $(\ee_{1},\ee_{2},\ee_{3})$, as follows:
\begin{equation}\label{eq:diagonal-antidaigonal}
  \bsigma^{d}=
  \begin{pmatrix}
    \sigma_{11}' & 0            & 0            \\
    0            & \sigma_{22}' & 0            \\
    0            & 0            & \sigma_{33}'
  \end{pmatrix},
  \qquad
  \bsigma^{\overline{d}}=
  \begin{pmatrix}
    0           & \sigma_{12} & \sigma_{13} \\
    \sigma_{12} & 0           & \sigma_{23} \\
    \sigma_{13} & \sigma_{23} & 0
  \end{pmatrix},
\end{equation}
where $ \sigma_{ij}'= \sigma_{ij}-\frac{1}{3} \sigma_{kk}\, \delta_{ij}$. The decomposition of $\bsigma$ into $\octa$-irreducible components (see~\cite{Bertram96}) writes
\begin{equation}\label{eq:decomposition}
  \bsigma =  \bsigma^{d} + \bsigma^{\overline{d}} +\frac{1}{3} (\tr\bsigma) \id ,
\end{equation}
where
\begin{equation*}
  \bsigma' = \bsigma - \frac{1}{3} (\tr\bsigma) \id = \bsigma^{d} + \bsigma^{\overline{d}},
\end{equation*}
is the usual deviatoric part of $\bsigma$. Note that the decomposition of $\bsigma$ into its three irreducible components is orthogonal (with respect to the canonical scalar product on $\Sym^{2}(\RR^{3})$). Let $\bP_{\octa}^{d}$ and $\bP_{\octa}^{\overline{d}}$ denote the orthogonal projectors of $\bsigma$ onto its irreducible components $\bsigma^{d}$ and $\bsigma^{\overline{d}}$ respectively. We have thus
\begin{equation}\label{eq:PdPdbar0}
  \bsigma^{d} := \bP_{\octa}^{d} : \bsigma \quad \text{and} \quad \bsigma^{\overline{d}} :=\bP_{\octa}^{\overline{d}}:\bsigma.
\end{equation}
where \cite{Ryc1984,Fra1995},
\begin{equation*}
  \bP_\octa^{\overline{d}} := \frac{1}{2} \sum_{i< j}\be_{ij}\otimes \be_{ij},
  \qquad
  \be_{ij} :=\ee_{i}\otimes\ee_{j}+\ee_{j}\otimes \ee_{i} \; (i\neq j),
\end{equation*}
and
\begin{equation*}
  \bP_\octa^{d}:=\bJ- \bP_\octa^{\overline{d}},
  \qquad \bJ=\bI-\frac{1}{3} \id \otimes \id,
\end{equation*}
where $\bI$ is the fourth-order identity tensor, with components $I_{ijkl}=\frac{1}{2}(\delta_{ik}\delta_{jl}+\delta_{il}\delta_{jk})$, and $\bJ$ the fourth-order deviatoric projector. However, the cubic material is not necessarily oriented along the directions of the canonical cubic basis $(\ee_{i})$. As a consequence the axes of the underlying cube are not necessarily the $\ee_{i}$ (the symmetry group is a conjugate $g \octa g^{-1}$ of canonical group $\octa$ for some $g\in \SO(3)$) and the cubic projectors $\bP_{\octa}^{d}$ and $\bP_{\octa}^{\overline{d}}$ have to be replaced by the projectors
\begin{equation}\label{eq:PdPdbar}
  \bP^{d} = g\star \bP_\octa^{d},
  \qquad
  \bP^{\overline{d}} = g\star \bP_\octa^{\overline{d}}.
\end{equation}
In terms of components, this writes
\begin{equation*}
  (\bP^{d})_{ijkl}=g_{ip}g_{jq}g_{kr}g_{ls} (\bP_\octa^{d})_{pqrs} \quad \text{and} \quad (\bP^{\overline{d}})_{ijkl}=g_{ip}g_{jq}g_{kr}g_{ls} (\bP_\octa^{\overline{d}})_{pqrs}.
\end{equation*}

\begin{rem}\label{rem:Kelvin-projectors}
  When the stiffness tensor $\bC$ (or its inverse, the compliance tensor $\bS$) of a cubic material is known in an arbitrary orthonormal basis, associated with group $g \octa g^{-1}$, one can obtain the cubic projectors $\bP^{d}$ and $\bP^{\overline{d}}$ without computing the rotation $g$. This task can be done thanks to \emph{Kelvin's spectral decomposition} of the matrix representation $\Sym^{2}(\Sym^{2}(\RR^{2}))$ of $\bS$, $\bP^{d}$ and $\bP^{\overline{d}}$ being then the so-called \emph{Kelvin's projectors} \cite{Ryc1984,Fra1995,DM2011}. This can also be done thanks to its harmonic decomposition \cite{MDC2019}.
\end{rem}

\begin{rem}
  The decomposition into $\octa$-irreducible components also applies to the second-order magnetostriction strain. It writes
  \begin{equation*}
    \bepsilon^{\mu}= \bepsilon^{\mu d}+\bepsilon^{\mu \overline{d}}+\frac{1}{3}\lambda_v\id,
    \qquad
    \begin{cases}
      \bepsilon^{\mu d} = \bP^{d}:\bepsilon^{\mu}
      \\
      \bepsilon^{\mu \overline{d}} = \bP^{\overline{d}}:\bepsilon^{\mu}
    \end{cases}
    ,
  \end{equation*}
  which details as
  \begin{equation*}
    \bepsilon^{\mu} = \frac{3}{2\,m_s^2} \left(\lambda_{100}(\mm \otimes \mm)^{d} +  \lambda_{111}(\mm \otimes \mm)^{\overline{d}} \right) +\frac{1}{3\,m_s^2}\lambda_v\tr (\mm \otimes \mm)\id,
  \end{equation*}
  where
  \begin{equation*}
    (\mm \otimes \mm)^{d} = \bP^{\overline{d}}:(\mm \otimes \mm), \qquad (\mm \otimes \mm)^{\overline{d}}=  \bP^{\overline{d}}:(\mm \otimes \mm),
  \end{equation*}
  $\lambda_{100}$, $\lambda_{111}$ are the magnetostriction constants (they are elongations measured along the $<100>$ and $<111>$ crystallographic directions during a magnetic loading of a perfect demagnetized single crystal) and $\lambda_{v}$ is the so-called volume magnetostriction. In the canonical cubic basis $(\ee_{i})$, one has $\bP^{d}=\bP_{\octa}^{d}$ and $\bP^{\overline{d}}=\bP_{\octa}^{\overline{d}}$, so that the usual expression \cite{Cul1972,DTdLac1993,Hubert2011}
  \begin{equation*}\label{eq:eps:mu}
    \bepsilon^{\mu} = \frac{3}{2}
    \begin{pmatrix}
      \lambda_{100}(\gamma_{1}^{2}-\frac{1}{3}) & \lambda_{111}\gamma_{1}\gamma_{2}         & \lambda_{111}\gamma_{1}\gamma_{3}         \\
      \lambda_{111}\gamma_{1}\gamma_{2}         & \lambda_{100}(\gamma_{2}^{2}-\frac{1}{3}) & \lambda_{111}\gamma_{2}\gamma_{3}         \\
      \lambda_{111}\gamma_{1}\gamma_{3}         & \lambda_{111}\gamma_{2}\gamma_{3}         & \lambda_{100}(\gamma_{3}^{2}-\frac{1}{3})
    \end{pmatrix}
    +\frac{1}{3}\lambda_v\id ,
  \end{equation*}
  is recovered. Note that the volume magnetostriction $\lambda_v$ is negligible for most metallic ferromagnetic materials (for which $\tr \bepsilon^{\mu} \approx 0$).
\end{rem}

\subsection{Polynomial invariants and minimal integrity basis}
\label{sec:invariant-algebra}

In practice, it is more useful to study \emph{polynomial invariant functions} by a group $G$, rather than just arbitrary invariant functions since algorithms are known to compute them. This is not a limitation if the energy density is a smooth function as it can be expanded to a given order into polynomials which inherit its symmetry. Real polynomial functions defined on a vector space $V$ form an \emph{algebra}, denoted by $\RR[V]$ (the sum and the product of polynomial functions are polynomial functions). Given a linear action of a group $G$ on $V$, the polynomial functions on $V$ which are invariant by $G$ is a \emph{subalgebra} of $\RR[V]$ (the sum and the product of invariant polynomial functions are still invariant polynomial functions), denoted by $\RR[V]^{G}$ and called the \emph{invariant algebra} of $V$ under $G$.

It is a fundamental result of invariant theory that, when the group $G$ is finite (like $\octa$) or compact (like $\OO(3)$ or $\SO(3)$), the invariant algebra is \emph{generated} by a finite number of invariant polynomials~\cite{Hilb1993,Stu1993}. This means that each $G$-invariant polynomial function can be rewritten as a polynomial function of these generators. Such a generating set is sometimes called an \emph{integrity basis}. A generating set is called \emph{minimal} if no proper subset of it is a generating set. It is always possible to find a minimal integrity basis made of homogeneous polynomials. A minimal homogeneous integrity basis is not unique but its cardinal and degrees of generators are uniquely defined.

\begin{rem}\label{rem:irreducible-invariants}
  An homogeneous polynomial invariant is called \textit{reducible} if it can be written as the product of two (non constant) homogeneous polynomial invariants, or more generally as a sum of products of two (non constant) homogeneous polynomial invariants. Otherwise, it is called \textit{irreducible}. A minimal integrity basis contains only irreducible invariants.
\end{rem}

\begin{exa}\label{ex:R3-plus-S2-O3}
  A minimal integrity basis of $\RR[V]^{\OO(3)}$, where $V = \RR^{3} \oplus \Sym^{2}(\RR^{3})$ is provided by
  \begin{equation*}
    \norm{\mm}^{2}, \quad \tr \bsigma, \quad \tr \bsigma^{\prime\, 2}, \quad \tr \bsigma^{\prime\, 3}, \quad \mm \cdot \bsigma' \mm, \quad \mm \cdot \bsigma^{\prime\, 2}\, \mm.
  \end{equation*}
\end{exa}

\begin{exa}\label{ex:R3-plus-S2-SO3}
  A minimal integrity basis of $\RR[V]^{\SO(3)}$, where $V = \RR^{3} \oplus \Sym^{2}(\RR^{3})$ is provided by
  \begin{equation*}
    \norm{\mm}^{2}, \quad \tr \bsigma, \quad \tr \bsigma^{\prime\, 2}, \quad \tr \bsigma^{\prime\, 3}, \quad \mm \cdot \bsigma' \mm, \quad \mm \cdot \bsigma^{\prime\, 2}\, \mm, \quad \det(\mm, \bsigma' \mm, \bsigma^{\prime\, 2} \, \mm).
  \end{equation*}
\end{exa}

The decomposition~\eqref{eq:decomposition} of $(\mm,\bsigma)$ into its irreducible components $(\mm,\bsigma^{d},\bsigma^{\overline{d}},\tr\bsigma)$ under the action of the cubic group $\octa$ allows to recast any polynomial function of $(\mm,\bsigma)$ as a polynomial function of its components $(\mm,\bsigma^{d},\bsigma^{\overline{d}},\tr\bsigma)$. The invariance of the free energy density by the cubic group $\octa$, writes
\begin{equation}\label{eq:inv:energie}
  \Psi^\star(g \star \mm,g \star \bsigma^{d},g \star \bsigma^{\overline{d}},g \star \tr\bsigma) = \Psi^\star(\mm,\bsigma^{d},\bsigma^{\overline{d}},\tr\bsigma),\qquad \forall g \in \octa .
\end{equation}

\begin{rem}\label{rem:reduction-process}
  Since $\tr\bsigma$ is an invariant, each invariant polynomial of $(\mm,\bsigma^{d},\bsigma^{\overline{d}},\tr\bsigma)$ which contains $\tr\bsigma$, but not reduced to $\tr\bsigma$ (up to a constant multiplicative factor), is reducible according to remark~\ref{rem:irreducible-invariants}. Since a minimal integrity basis contains only irreducible invariants, a minimal integrity basis of $\RR[V]^{\octa}$ is made of $\tr\bsigma$ together with a minimal integrity basis of the invariant algebra of $\RR[V^{\prime}]^{\octa}$, where $V^{\prime}$ is the vector space spanned by $(\mm,\bsigma^{d},\bsigma^{\overline{d}})$.
\end{rem}

\subsection{Counting the number of independent multi-homogeneous invariants}
\label{sec:Hilbert-series}

A very useful tool in invariant theory is the \emph{Hilbert series}~\cite[section 2.2]{Stu1993} (see also~\cite[Spencer, p. 181--184]{Boe1987}, \cite{NS2002} and \cite{Smi1960,Smi1965,Spe1970a}). Indeed any polynomial invariant can be decomposed uniquely as a sum of \emph{homogeneous polynomial invariants}. In other words, the invariant algebra $\RR[V]^{G}$ can be written as the direct sum
\begin{equation*}
  \RR[V]^{G} = \bigoplus_{n=0}^{\infty} \RR[V]_{n}^{G},
\end{equation*}
where $\RR[V]_{n}^{G}$ is the finite dimensional vector space of homogeneous polynomial invariants of degree $n$. This makes the invariant algebra a \emph{graded algebra} (by the total degree). The Hilbert series encodes the dimension $a_n= \dim \RR[V]^{G}_{n}$ of the different finite dimensional vector spaces $\RR[V]^{G}_{n}$ and is defined as
\begin{equation*}
  H_{\rho}(z) = \sum_{n} a_{n} \, z^{n} ,
  \quad \text{with} \quad  a_{n} = \dim \RR[V]^{G}_{n} .
\end{equation*}
Thus, each coefficient $a_{n}$ is the number of linearly independent homogeneous invariants of degree $n$. The remarkable fact is that the Hilbert series can be computed \textit{a priori}, using the Molien-Weyl formula for finite groups~\cite[Theorem 2.2.1 p.29]{Stu1993}
\begin{equation}\label{eq:Molien-Weyl}
  H_{\rho}(z) = \frac{1}{\abs{G}}\sum_{g \in G} \frac{1}{\det (I-z\rho(g))},
\end{equation}
where $\rho(g)$ is the linear mapping on $V$ defined by
\begin{equation*}
  \rho(g)(\mm,\bsigma) = (g\mm,g \bsigma g^{t}).
\end{equation*}

There are moreover several useful \emph{refinements} of the Hilbert series, when the space $V$ can be decomposed into stable subspaces. For instance, $V = \RR^{3} \oplus \Sym^{2}(\RR^{3})$ ($\mm\in \RR^{3}$, $\bsigma\in \Sym^{2}(\RR^{3})$) splits naturally into two stable vector spaces, one defined by the vector $\mm$ and another defined by the stress $\bsigma$. In that case, any polynomial invariant can be decomposed into invariant polynomials which are simultaneously homogeneous in $\mm$ and $\bsigma$. In other words, this splitting leads to a new decomposition of the invariant algebra $\RR[V]^{G}$ into the direct sum
\begin{equation*}
  \RR[V]^{G} = \bigoplus_{\alpha,\beta} \RR[V]_{\alpha\beta}^{G},
\end{equation*}
where $\RR[V]_{\alpha,\beta}^{G}$ is the finite dimensional vector space of bi-homogeneous polynomial invariants of degree $\alpha$ in $\mm$ and degree $\beta$ in $\bsigma$. To this \emph{bi-graded algebra} structure of $\RR[V]^{G}$ corresponds the \emph{two-variable Hilbert series}
\begin{equation*}
  H_{\rho}(z_{m},z_{\sigma}) = \sum_{\alpha,\beta} a_{\alpha\beta} \, z_{m}^{\alpha}z_{\sigma}^{\beta} ,
  \quad \text{where} \quad
  a_{\alpha\beta} = \dim \RR[V]^{G}_{\alpha\beta} .
\end{equation*}
and the Molien-Weyl formula writes
\begin{equation*}
  H_{\rho}(z_{m},z_{\sigma}) = \frac{1}{\abs{G}}\sum_{g \in G} \frac{1}{\det (I-z_{m}\rho_{m}(g))}\frac{1}{\det (I-z_{\sigma}\rho_{\sigma}(g))},
\end{equation*}
where $\rho_{m}$ is the matrix representation of $G$ on $\mm$ (considered here as a vector) and $\rho_{\sigma}$ is the matrix representation on $\bsigma$,
\begin{equation*}
  \rho_{m}(g) \,\mm=g\star \mm=g\,\mm,
  \qquad
  \rho_{\sigma}(g) \,\bsigma=g\star  \bsigma=g\, \bsigma g^t.
\end{equation*}
For the proper octahedral group $G= \octa^+$, we obtain
\begin{equation}\label{eq:Hilbert-octa-plus-bi-graded}
  \begin{aligned}
    H_{(V,\octa^+)}(z_{m},z_{\sigma}) & = \frac{1}{24} \frac{1}{1-z_{\sigma}} \left\{ \frac{1}{(1-z_{m})^{3}(1-z_{\sigma})^5} + \frac{6}{(1-z_{m})(1+z_{m})^{2}(1-z_{\sigma}^{2})^{2}(1-z_{\sigma})}\right.
    \\
                                      & \quad + \left. \frac{3}{(1+z_{m})(1-z_{m}^{2})(1-z_{\sigma})(1-z_{\sigma}^{2})^{2}} + \frac{8(1-z_{\sigma})}{(1-z_{m}^{3})(1-z_{\sigma}^{3})^{2}} \right .
    \\
                                      & \quad + \left.\frac{6(1-z_{\sigma})}{(1-z_{m})(1+z_{m}^{2})(1-z_{\sigma}^{2})(1-z_{\sigma}^{4})}\right\}
  \end{aligned}
\end{equation}
whereas for the full octahedral group $G=\octa$, we have
\begin{equation*}
  H_{(V,\octa)}(z_{m},z_{\sigma}) = \frac{1}{2} \left(H_{(V,\octa^+)}(z_{m},z_{\sigma})+H_{(V,\octa^+)}(-z_{m},z_{\sigma})\right).
\end{equation*}

\autoref{tab:coeffs} summarizes for the cubic invariant algebra $\RR[V]^\octa$ ($\rho=(V, \octa)$) the number of linearly independent bi-homogeneous invariants for different bi-degrees  in $\mm$ and $\bsigma$. We recover in particular the fact that there are three (so-called purely elastic) invariants of bi-degree $(0,2)$, \emph{i.e.} independent of $\mm$ and quadratic in the symmetric second-order tensor $\bsigma$.

\begin{table}[ht]
  \centering
  \begin{tabular}{cccccccccccc}
    \toprule
    \diagbox{$\mm$}{$\bsigma$} & 0 & 1  & 2  & 3   & 4   & 5   & 6    & 7    & 8    & 9    & 10
    \\
    \midrule
    0                          & 1 & 1  & 3  & 6   & 11  & 18  & 32   & 48   & 75   & 111  & 160
    \\
    2                          & 1 & 3  & 9  & 20  & 42  & 78  & 138  & 228  & 363  & 553  & 819
    \\
    4                          & 2 & 6  & 19 & 44  & 95  & 180 & 323  & 540  & 867  & 1330 & 1980
    \\
    6                          & 3 & 10 & 32 & 78  & 168 & 324 & 585  & 984  & 1584 & 2442 & 3640
    \\
    8                          & 4 & 15 & 49 & 120 & 263 & 510 & 963  & 1560 & 2517 & 3885 & 5802
    \\
    10                         & 5 & 21 & 69 & 172 & 378 & 738 & 1338 & 2268 & 3663 & 5663 & 8463
    \\
    \bottomrule
  \end{tabular}
  \caption{Number of linearly independent bi-homogeneous invariants for different bi-degrees in $(\mm, \bsigma)$ and for $G=\octa$.}
  \label{tab:coeffs}
\end{table}

There is finally a third variation of the Hilbert series that will be useful for us and concerns the space $V^{\prime}=\RR^{3} \oplus \HH^{2}(\RR^{3})$, to which belongs the pair $(\mm, \bsigma')$ and where $\HH^{2}(\RR^{3})$ is the vector space of deviatoric second-order tensors. It is formulated using the \emph{tri-graduation} of $\RR[V^{\prime}]^{G}$ induced by the $\octa$-stable decomposition of $V^{\prime}$ into the variables $(\mm,\bsigma^{d},\bsigma^{\overline{d}})$. This makes the invariant algebra $\RR[V^{\prime}]^{G}$ into a \emph{tri-graded algebra}, \textit{i.e.} each invariant polynomial $p(\mm,\bsigma^{d},\bsigma^{\overline{d}})$ can be decomposed uniquely into the sum of polynomials which are multihomogeneous in $\mm$, $\bsigma^{d}$ and $\bsigma^{\overline d}$ (multihomogeneous simultaneously in all three variables). For instance, the invariant $I_{210}$ in~\autoref{tab:even-invariants} writes
\begin{equation*}
  I_{210} = (\mm \otimes \mm)^{d} : \bsigma^{d} = m_{1}^{2}\sigma_{11}' + m_{2}^{2}\sigma_{22}' + m_{3}^{2} \sigma_{33}',
\end{equation*}
with $\sigma_{11}'+ \sigma_{22}'+ \sigma_{33}'=0$. It is homogeneous of degree $2$ in $\mm$, of degree $1$ in $\bsigma^{d}$ and of degree $0$ in $\bsigma^{\overline{d}}$ and its total degree is $3$. The invariant algebra $\RR[V^{\prime}]^{G}$ decomposes thus into the direct sum
\begin{equation*}
  \RR[V^{\prime}]^{G} = \bigoplus _{\alpha, \beta, \gamma} \RR[V^{\prime}]^{G}_{\alpha \beta \gamma},
\end{equation*}
where $\RR[V^{\prime}]^{G}_{\alpha \beta \gamma}$ is the finite dimensional vector space of polynomial invariants which are homogeneous of degree $\alpha$ in $\mm$, $\beta$ in $\bsigma^{d}$ and $\gamma$ in $\bsigma^{\overline{d}}$. The corresponding Hilbert series writes
\begin{equation*}
  H_{\rho}(z_{m},z_{d},z_{\bar d}) = \sum a_{\alpha \beta \gamma} \, {z_{m}}^{\alpha} {z_{d}}^{\beta} {z_{\bar d}}^{\gamma} ,
  \quad \text{where} \quad
  a_{\alpha \beta \gamma}=\dim \RR[V^{\prime}]^{G}_{\alpha \beta \gamma} ,
\end{equation*}
and the Molien-Weyl formula gives us
\begin{equation*}
  H_{\rho}(z_{m},z_{d},z_{\bar d}) = \frac{1}{\abs{G}}\sum_{g \in G} \frac{1}{\det (I-z_{m}\rho_{m}(g))}\frac{1}{\det (I-z_{d}\rho_{d}(g))} \frac{1}{\det (I-z_{\bar d}\rho_{\bar d}(g))},
\end{equation*}
where, in addition to $\rho_{m}$, $\rho_{d}$ is the matrix representation of $G$ on $\bsigma^{d}$, and $\rho_{\bar d}$ on $\bsigma^{\overline{d}}$,
\begin{equation*}
  \rho_{d}(g) \,\bsigma^{d} = g\star \bsigma^{d} = g\, \bsigma^{d} g^t,
  \qquad
  \rho_{\bar d}(g) \,\bsigma^{\overline{d}} = g\star \bsigma^{\overline{d}} = g \,\bsigma^{\overline{d}} g^t.
\end{equation*}
For $G=\octa^+$, we get
\begin{equation}\label{eq:Hilbert-octa-plus-tri-graded}
  \begin{aligned}
    H_{(V',\octa^+)}(z_{m},z_{d},z_{\bar d}) & = \frac{1}{24} \left( \frac{1}{(1-z_{m})^{3}(1-z_{d})^{2}(1-z_{\bar d})^{3}} + \frac{6}{(1-z_{m})(1+z_{m})^{2}(1-z_{d}^{2})(1+z_{\bar d})(1-z_{\bar d})^{2}} \right.
    \\
                                             & \quad + \left.\frac{3}{(1+z_{m})(1-z_{m}^{2})(1-z_{d})^{2}(1+z_{\bar d})(1-z_{\bar d}^{2})} + \frac{8(1-z_{d})}{(1-z_{m}^{3})(1-z_{d}^{3})(1-z_{\bar d}^{3})} \right .
    \\
                                             & \quad + \left.\frac{6(1-z_{\bar d})}{(1-z_{m})(1+z_{m}^{2})(1-z_{d}^{2})(1-z_{\bar d}^{4})}\right)
  \end{aligned}
\end{equation}
and for $G=\octa$, we have
\begin{equation}\label{eq:Hilbert-octa-tri-graded}
  H_{(V',\octa)}(z_{m},z_{d},z_{\bar d})  = \frac{1}{2} \left(H_{(V',\octa^+)}(z_{m},z_{d},z_{\bar d}) +H_{(V',\octa^+)}(-z_{m},z_{d},z_{\bar d}) \right).
\end{equation}

\begin{rem}
  The coefficients $a_{\alpha \beta \gamma}$ which are the dimensions of the vector spaces $\RR[V^{\prime}]^{G}_{\alpha \beta \gamma}$ are then computed using Taylor's expansion of the rational functions~\eqref{eq:Hilbert-octa-plus-tri-graded} and~\eqref{eq:Hilbert-octa-tri-graded}.
\end{rem}

\subsection{Cubic minimal integrity bases}
\label{sec:integrity-basis}

Generating sets for the invariant algebras $\RR[V]^\octa$ and $\RR[V]^{\octa^{+}}$ for $V=\RR^{3}\oplus \Sym^{2}(\RR^{3})$ seem to have been first proposed by Smith, Smith and Rivlin in 1963~\cite{SSR1963}. The set of generators produced therein for the orientation preserving cubic subgroup $\octa^{+}$ is however not minimal as explained below. To produce a generating set of invariants is one thing, which can be achieved by clever use of mathematics. Verifying that a generating set is minimal requires powerful computation tools (see~\autoref{sec:algo}), and it is easy to \emph{miss certain relations}, as acknowledged by Spencer himself in~\cite[Chapter 8, pages 163-164]{Boe1987}, and thus produce a wrong answer if a systematic and lengthy calculation is not carried out carefully to the end. This explains why an error has been found in~\cite{SSR1963} (one invariant is indeed reducible). Computation tools in the sixties were certainly not sufficient. Another weakness of~\cite{SSR1963} is that generators are expressed using components $\sigma_{ij}$, $m_{i}$ of $\bsigma$ and $\mm$. The complexity of the formulas increase rapidly and it is not always easy to calculate the constitutive model response with an energy density expressed in expanded components.

In the present work, we propose new sets of generators, which are shown to be moreover minimal. In addition, these generators are written in intrinsic form (not in components). They are detailed in \autoref{tab:even-invariants} and \autoref{tab:odd-invariants} where the notations $\ba^{d\,n}:=(\ba^{d})^{n}$ and $\ba^{\overline d\,n}:=(\ba^{\overline d})^{n}$ have been used. Their expressions are thus much more compact and allow calculations --- of their partial derivatives for instance, as in next section --- to be much easier. To obtain these formulas, we have used only two fundamental tensorial operations. The first one is the classical \emph{contraction between two tensors}, on two or more indices. The second one is the \emph{generalized cross product} between two totally symmetric tensors of any orders, which was introduced in~\cite{OKDD2018} (see also \cite{ADDKO2020} for practical examples). Given $\bm{S}^{1} \in \mathbb{S}^p(\RR^{3})$ and $\bm{S}^{2} \in \mathbb{S}^p(\RR^{3})$, it is defined by
\begin{equation}\label{eq:GenCross}
  \bm{S}^{1} \times \bm{S}^{2} = -(\bm{S}^{1} \cdot \bm{\varepsilon} \cdot \bm{S}^{2})^s \in \mathbb{S}^{p+q-1}(\RR^{3}),
\end{equation}
where $(\cdot)^s$ means the total symmetrization and where $\bm{\varepsilon}$ is the \textit{Levi-Civita} third order tensor ($\varepsilon_{ijk}=\det(\ee_{i}, \ee_{j},\ee_{k})$ in any direct orthonormal basis $\ee_{i}$). Using these operations, has been produced a new set of 60 $\octa^{+}$-invariants which are given
in~\autoref{tab:even-invariants} and~\autoref{tab:odd-invariants}. The proofs of the following two theorems are provided in~\autoref{sec:proofs}.

\begin{thm}\label{thm:main-octa-plus}
  The 60 invariants in~\autoref{tab:even-invariants} and~\autoref{tab:odd-invariants} form a \emph{minimal integrity basis} of
  \begin{equation*}
    \RR[\RR^{3} \oplus \Sym^{2}(\RR^{3})]^{\octa^{+}},
  \end{equation*}
  \emph{i.e.} for the orientation-preserving subgroup $\octa^{+}\subset \octa$.
\end{thm}

\begin{rem}
  Polynomials in~\autoref{tab:even-invariants} are invariant by the full octahedral group $\octa$, and thus by its subgroup $\octa^{+}$. Polynomials in~\autoref{tab:odd-invariants} are only invariant by $\octa^{+}$ and change sign when the central symmetry $I_{c}$ is applied to them. They are indeed odd in magnetization $\mm$ whereas $I_{c}\star\mm=-\mm$.
\end{rem}

\begin{thm}\label{thm:main-octa}
  The 29 invariants in~\autoref{tab:even-invariants}, altogether with  $(J_{111})^{2}$ where
  \begin{equation*}
    J_{111} = \tr \big( (\bsigma^{d} \times \mm)\bsigma^{\overline{d}} \big)
  \end{equation*}
  form a minimal integrity basis of $\RR[\RR^{3} \oplus \Sym^{2}(\RR^{3})]^{\octa}$, \emph{i.e.} for the full octahedral group $\octa$.
\end{thm}

\begin{rem}\label{rem:I222}
  The above expression of $(J_{111})^{2}$ makes use of the generalized cross product \eqref{eq:GenCross}. It can advantageously be replaced by another invariant that is expressed only through tensors contraction. For instance, one can use instead
  \begin{equation}\label{eq:I222}
    I_{222}:=(\mm\otimes \mm)^{\overline{d}} :\big( \bsigma^{d} (\bsigma^{\overline{d}\, 2})^{\overline{d}} \bsigma^{d}\big),
  \end{equation}
  of multi-degree $(2,2,2)$ in $(\mm, \bsigma^{d}, \bsigma^{\overline{d}})$. Indeed,
  \begin{equation}\label{J111square}
    (J_{111})^{2}  = -3 I_{222}+\frac{1}{12}I_{020}I_{002}I_{200}-\frac{3}{2}I_{020}I_{202}^{a}-I_{020}I_{202}^{b}-\frac{1}{2}I_{002}I_{220}+I_{012}I_{210}-\frac{1}{2}I_{022}I_{200}.
  \end{equation}
\end{rem}

\begin{table}[ht]
  \centering
  \begin{tabular}{ccccc}
  \toprule
     & deg($\mm$) & deg($\bsigma$) & Expression                                                                                                                    & Tri-graded notation \\
  \midrule
  1  & 0          & 1              & \footnotesize $\tr\bsigma$                                                                                                 & --          \\
  2  & 0          & 2              & \footnotesize $\bsigma^{\bar d} : \bsigma^{\bar d}$                                                                        & $I_{002}$           \\
  3  & 0          & 2              & \footnotesize $\bsigma^{d}: \bsigma^{d}$                                                                                   & $I_{020}$           \\
  4  & 0          & 3              & \footnotesize $\tr(\bsigma^{\bar d\, 3})$                                                                               & $I_{003}$           \\
  5  & 0          & 3              & \footnotesize $\bsigma^{\bar d\, 2} :\bsigma^{d}$                                                                          & $I_{012}$           \\
  6  & 0          & 3              & \footnotesize $\tr(\bsigma^{d\,3})$                                                                                    & $I_{030}$           \\
  7  & 0          & 4              & \footnotesize $ (\bsigma^{\bar d\, 2})^{\bar d} : ( \bsigma^{\bar d\, 2})^{\bar d}$
     & $I_{004}$       \\
  8  & 0          & 4              & \footnotesize $ \tr( \bsigma^{\bar d}  \bsigma^{d} \bsigma^{\bar d}\bsigma^{d} )$            & $I_{022}$           \\
  9  & 0          & 5              & \footnotesize $\big( \bsigma^{\bar d} (\bsigma^{\bar d\, 2})^{\bar d} \bsigma^{\bar d}\big):\bsigma^{d}$                & $I_{014}$           \\
  10  & 2          & 0              & \footnotesize $ \Vert \mm \Vert^{2}$                                                                                       & $I_{200}$           \\
  11 & 2          & 1              & \footnotesize $(\mm \otimes \mm)^{\bar d}: \bsigma^{\bar d}$                                          & $I_{201}$           \\
  12 & 2          & 1              & \footnotesize $ (\mm \otimes \mm)^d : \bsigma^{d} $                                                                        & $I_{210}$           \\
  13 & 2          & 2              & \footnotesize $(\mm \otimes \mm)^{d}: \bsigma^{\bar d\, 2} $                                                               & $I_{202}^{a}$       \\
  14 & 2          & 2              & \footnotesize $ (\mm \otimes \mm)^{\bar d}:\bsigma^{\bar d\, 2} $                                                          & $I_{202}^{b}$       \\
  15 & 2          & 2              & \footnotesize $(\mm \otimes \mm)^{\bar d} :(\bsigma^{\bar d}\bsigma^{d})$                                                  & $I_{211}$           \\
  16 & 2          & 2              & \footnotesize $(\mm \otimes \mm)^d: \bsigma^{d\, 2}$                                                                       & $I_{220}$           \\
  17 & 2          & 3              & \footnotesize $(\mm \otimes \mm)^{\bar d}:\big((\bsigma^{\bar d\, 2})^{\bar d}\bsigma^{\bar d} \big)$                      & $I_{203}$           \\
  18 & 2          & 3              & \footnotesize $(\mm \otimes \mm)^d:\big((\bsigma^{\bar d\, 2})^{d}\bsigma^{d} \big)$                                       & $I_{212}^{a}$       \\
  19 & 2          & 3              & \footnotesize $(\mm \otimes \mm)^{\bar d}:\big((\bsigma^{\bar d\, 2})^{\bar d}\bsigma^{d} \big)$                           & $I_{212}^{b}$       \\
  20 & 2          & 3              & \footnotesize $(\mm \otimes \mm)^{\bar d}:\big(\bsigma^{d}\bsigma^{\bar d}\bsigma^{d}\big)$                                & $I_{221}$           \\
  21 & 2          & 4              & \footnotesize $(\mm \otimes \mm)^{d}:\big(\bsigma^{\bar d}(\bsigma^{\bar d\, 2})^{\bar d}\bsigma^{\bar d}\big)$            & $I_{204}$           \\
  22 & 2          & 4              & \footnotesize $(\mm \otimes \mm)^{\bar d}: \big((\bsigma^{\bar d\, 2})^{d} \bsigma^{\bar d }\bsigma^{d} \big)$             & $I_{213}$           \\
  23 & 4          & 0              & \footnotesize $ (\mm \otimes \mm)^{\bar d}:(\mm \otimes \mm)^{\bar d}$                                                     & $I_{400}$           \\
  24 & 4          & 1              & \footnotesize $ (\mm \otimes \mm)^{\bar d\, 2}:\bsigma^{\bar d} $                    & $I_{401}$           \\
  25 & 4          & 1              & \footnotesize $(\mm \otimes \mm)^{\bar d\, 2}:\bsigma^{d} $                         & $I_{410}$           \\
  26 & 4          & 2              & \footnotesize $(\mm \otimes \mm)^{\bar d\, 2}:\big(\bsigma^{\bar d\, 2}\big)^{\bar d}$     & $I_{402}$           \\
  27 & 4          & 2              & \footnotesize $(\mm \otimes \mm)^{\bar d\, 2}:\big(\bsigma^{d}\bsigma^{\bar d}\big)$         & $I_{411}$           \\
  28 & 6          & 0              & \footnotesize $ \tr\big( (\mm \otimes \mm)^{\bar d\, 3} \big)$                                                       & $I_{600}$           \\
  29 & 6          & 1              & \footnotesize $\tr \big( (\mm \otimes \mm)^{d}(\mm \otimes \mm)^{\bar d} (\mm \otimes \mm)^{d}\bsigma^{\bar d}  \big)$ & $I_{601}$           \\
  \bottomrule
\end{tabular}

  \caption{$\octa^{+}$-invariants without cross product ($\textrm{deg}(\mm)=\alpha$,  $\textrm{deg}(\bsigma)=\beta+\gamma$ except line $1$).}
  \label{tab:even-invariants}
\end{table}

\begin{table}[ht]
  \centering
  \begin{tabular}{ccccc}
  \toprule
     & deg($\mm$) & deg($\bsigma$) & Expression                                                                                                                                                                                    & Tri-graded notation
  \\
  \midrule
  30 & 1          & 2              & \footnotesize $(\bsigma^{d} \times \mm):\bsigma^{\bar d}$                                                                                                                & $J_{111}$
  \\
  31 & 1          & 3              & \footnotesize $ \big((\bsigma^{\bar d\, 2})^{\bar d} \times \mm\big):\bsigma^{\bar d} $                                                                                          & $J_{103}$           \\
  32 & 1          & 3              & \footnotesize $ \big(\bsigma^{\bar d\, 2} \times \mm\big):\bsigma^{d} $                                                                                                          & $J_{112}$
  \\
  33 & 1          & 3              & \footnotesize $\big(\bsigma^{\bar d}\times (\bsigma^{d}\mm)\big):\bsigma^{d}$                                                                                                    & $J_{121}$
  \\
  34 & 1          & 4              & \footnotesize $\big(\bsigma^{\bar d} \times (\bsigma^{\bar d}\mm)\big)^{\bar d}:\bsigma^{\bar d\, 2} $                                                                            & $J_{104}$           \\
  35 & 1          & 4              & \footnotesize $\big(((\bsigma^{\bar d\, 2})^{d}\bsigma^{d}) \times \mm\big):\bsigma^{\bar d}$                                                                                    & $J_{113}^{a}$
  \\
  36 & 1          & 4              &  \footnotesize $ \big( (\bsigma^{\bar d\, 2})^{\bar d} \times (\bsigma^{d}\mm)\big):\bsigma^{\bar d}$                                                                              & $J_{113}^{b}$       \\
  37 & 1          & 4              & \footnotesize $\big((\bsigma^{\bar d\, 2})^{\bar d} \times (\bsigma^{d}\mm)\big):\bsigma^{d}$                                                                                   & $J_{122}$           \\
  38 & 1          & 4              & \footnotesize $ \big(\bsigma^{d\, 2} \times(\bsigma^{d}\mm)\big): \bsigma^{\bar d}$                                                                                 & $J_{131}$           \\
  39 & 1          & 5              & \footnotesize $\big((\bsigma^{\bar d\, 2})^{\bar d} \times (\bsigma^{\bar d}\mm)\big):\bsigma^{\bar d\, 2} $                                                           & $J_{105}$           \\
  40 & 1          & 5              & \footnotesize $\big(  (\bsigma^{\bar d\, 2})^{\bar d} \times (\bsigma^{d}\bsigma^{\bar d} \mm) \big): \bsigma^{\bar d}$                                                        & $J_{114}$           \\
  41 & 1          & 5              & \footnotesize $\left[\big(\bsigma^{d} (\bsigma^{\bar d\, 2})\bsigma^{ d}\big)^{\bar d} \times  \mm \right]:\bsigma^{\bar d}$                                         & $J_{123}$           \\
  42 & 1          & 5              & \footnotesize $\big(\bsigma^{\bar d\, 2} \times (\bsigma^{d}\mm) \big):\bsigma^{d\, 2}$                                                                          & $J_{132}$           \\
  43 & 3          & 1              & \footnotesize $(\mm \otimes \mm)^{\bar d}:(\bsigma^{\bar d} \times \mm) $                                                                                                 & $J_{301}$           \\
  44 & 3          & 2              & \footnotesize $(\mm \otimes \mm)^{d} :(\bsigma^{\bar d} \times (\bsigma^{\bar d}\mm))$                                                                                    & $J_{302}$           \\
  45 & 3          & 2              & \footnotesize $(\mm \otimes \mm)^{\bar d} :( \bsigma^{\bar d} \times (\bsigma^{d}\mm))$                                                                                   & $J_{311}^{a}$       \\
  46 & 3          & 2              & \footnotesize $(\mm \otimes \mm)^d:\big( \bsigma^{\bar d} \times (\bsigma^{d}\mm) \big)$                                                                                & $J_{311}^{b}$       \\
  47 & 3          & 3              & \footnotesize $(\mm \otimes \mm)^d :\big(\bsigma^{\bar d} \times \big((\bsigma^{\bar d\, 2})^{\bar d} \mm \big))$                                                          & $J_{303}^{a}$       \\
  48 & 3          & 3              & \footnotesize $(\mm \otimes \mm)^{\bar d} :\left[\bsigma^{\bar d} \times \big((\bsigma^{\bar d\, 2})^{\bar d} \mm \big)\right]$                                                   & $J_{303}^{b}$       \\
  49 & 3          & 3              & \footnotesize $  \big((\mm \otimes \mm)^{\bar d} \times (\bsigma^{d}\mm) \big):\big(\bsigma^{\bar d \,2}\big)^{d}$                                                          & $J_{312}^{a}$       \\
  50 & 3          & 3              & \footnotesize $\big((\mm \otimes \mm)^{\bar d} \times (\bsigma^{d}\mm) \big):\big(\bsigma^{\bar d \,2}\big)^{\bar d}$                                                       & $J_{312}^{b}$       \\
  51 & 3          & 3              & \footnotesize $\left[ (\bsigma^{d}(\mm \otimes \mm)^{d}) \times (\bsigma^{d}\mm) \right]:\bsigma^{\bar d}$                                                & $J_{321}$           \\
  52 & 3          & 3              & \footnotesize $ \big((\mm \otimes \mm)^{\bar d} \times (\bsigma^{d\, 2}\mm) \big):\bsigma^{d}$                                                                         & $J_{330}$           \\
  53 & 5          & 1              & \footnotesize $ \big[(\mm \otimes \mm)^{\bar d} \times ((\mm \otimes \mm)^{\bar d}\mm)\big]:\bsigma^{\bar d}$                                                                     & $J_{501}$           \\
    54 & 5          & 1              & \footnotesize $\big[(\mm \otimes \mm)^{\bar d} \times ((\mm \otimes \mm)^{\bar d}\mm ) \big]:\bsigma^{d}$                                                   & $J_{510}$           \\
  55 & 5          & 2              & \footnotesize $\big[(\mm \otimes \mm)^d \times ((\mm \otimes \mm)^{\bar d}\bsigma^{\bar d}\mm ) \big]:\bsigma^{\bar d}$                                     & $J_{502}$           \\
  56 & 5          & 2              & \footnotesize $\big[(\mm \otimes \mm)^d \times ((\mm \otimes \mm)^{\bar d}\bsigma^{d}\mm ) \big]:\bsigma^{\bar d}$                                           & $J_{511}$           \\
  57 & 5          & 2              & \footnotesize $\big[(\mm \otimes \mm)^{\bar d} \times ((\mm \otimes \mm)^{\bar d}\bsigma^{d}\mm ) \big]:\bsigma^{d}$                                        & $J_{520}$           \\
  58 & 7          & 1              & \footnotesize $\big[(\mm \otimes \mm)^d \times \big(((\mm \otimes \mm)^{\bar d}(\mm \otimes \mm)^{\bar d})^{\bar d}\mm\big) \big]:\bsigma^{\bar d}$    & $J_{701}$           \\
  59 & 7          & 1              & \footnotesize $\big[(\mm \otimes \mm)^{\bar d} \times \big(((\mm \otimes \mm)^{\bar d}(\mm \otimes \mm)^{\bar d})^{ d}\mm\big) \big]:\bsigma^{ d}$ & $J_{710}$           \\
  60 & 9          & 0              & \footnotesize $\big[(\mm \otimes \mm)^{\bar d} \times \big(((\mm \otimes \mm)^{\bar d}(\mm \otimes \mm)^{\bar d})^{ d}\mm\big) \big]:(\mm \otimes \mm)^{ d}$                                       & $J_{900}$           \\
  \bottomrule
\end{tabular}

  \caption{$\octa^{+}$-invariants with cross product ($\textrm{deg}(\mm)=\alpha$,  $\textrm{deg}(\bsigma)=\beta+\gamma$).}
  \label{tab:odd-invariants}
\end{table}

\begin{rem}
  In~\autoref{tab:even-smith-invariants} and~\autoref{tab:odd-smith-invariants}, we have translated the integrity bases for $\octa^{+}$ and $\octa$ provided by Smith et al in the new integrity basis provided by~\autoref{tab:even-invariants} and~\autoref{tab:odd-invariants}. The original set of generators proposed by Smith et al is expressed using components of the tensors with the following notations: $E_{i}=m_{i}$ and $g_{ij}=\sigma_{ij}$. The notations $I_{k}$, $L_{k}$ and $J_{k}K_l$ used in the tables are those introduced in Smith--Smith--Rivlin paper~\cite[Section 6]{SSR1963}. In tables \ref{tab:even-smith-invariants} and \ref{tab:odd-smith-invariants}, all Smith--Smith--Rivlin invariants restricted to $(\bsigma',\mm)$ are expressed in terms of $I_{\alpha\beta\gamma}$ and $J_{\alpha\beta\gamma}$.
\end{rem}

The minimal number of generators of $\RR[V]^{\octa^{+}}$ is 60, \textit{i.e.} one less than the number of generators proposed by Smith et al in 1963. One of these invariants is therefore reducible. In order to check which one rewrites polynomially as function of the others, the algorithm in~\autoref{sec:algo} has been applied to the list $\mathcal L$ of $\octa^{+}$-invariants provided in~\autoref{tab:even-smith-invariants} and~\autoref{tab:odd-smith-invariants} (with the bound $N=12$ for the highest total degree to be checked, see theorem~\ref{thm:octa-plus-bound}). As a result, a basis of the vector space of homogeneous invariants of multi-degree $(2,2,2)$ in $(\mm, \bsigma^{d}, \bsigma^{\overline{d}})$ is spanned by
\begin{equation*}
  I_{2}I_4I_{10} ,\quad  I_{2}I_{15} , \quad I_{2}I_{16} , \quad I_4I_{18} , \quad I_{7}I_{14} , \quad I_9I_{10} , \quad (L_{1})^{2}.
\end{equation*}
where $I_k$ are the Smith--Smith--Rivlin invariants of \autoref{tab:even-smith-invariants}. Indeed, $I_{25}$ can be recast as a function of the other Smith et al invariants as
\begin{equation*}
  I_{25} = \frac{1}{6} \left( I_{2}I_4I_{10} - 3I_{2}I_{15} + 4I_{2}I_{16} - I_4I_{18} - I_{7}I_{14} - I_9I_{10} - (L_{1})^{2} \right).
\end{equation*}

\begin{rem}
  In~\autoref{tab:even-invariants} and~\autoref{tab:odd-invariants}, almost all invariants have different multi-degrees $(\alpha, \beta, \gamma)$ in $(\mm, \bsigma^{d}, \bsigma^{\overline d})$. When there is no ambiguity and only one invariant of multi-degree $(\alpha, \beta, \gamma)$, it is denoted as $I_{\alpha \beta \gamma}$ (\autoref{tab:even-invariants}) and $J_{\alpha \beta \gamma}$ (\autoref{tab:odd-invariants}) in the fifth column. However, some pairs of invariants have the same multi-degree. This is the case for the pairs of lines $(13,14)$, $(18,19)$, $(35,36)$, $(45,46)$, $(47,48)$, and $(49,50)$. An exponent $^{a}$ or $^{b}$ has been added for a clear distinction. For instance, in~\autoref{tab:even-invariants}, the invariant in line 13 has been denoted by $I^{a}_{202}$; it has been denoted $I^{b}_{202}$ in line $14$.
\end{rem}

\begin{table}[ht]
  \centering
  \begin{tabular}{ccccc}
  \toprule
     & deg($\mm$) & deg($\bsigma$) & Evaluation                                                                                                                    & Smith et al Notation
  \\
  \midrule
  1  & 0          & 1              & \footnotesize $ \tr \bsigma $                                                                                                  & $I_{1}$
  \\
  2  & 0          & 2              & \footnotesize $-\frac{1}{2} I_{020} $                                                                                      & $I_{2}$
  \\
  3  & 0          & 2              & \footnotesize $\frac{1}{2} I_{002} $                                                                                       & $I_{4}$
  \\
  4  & 0          & 3              & \footnotesize $\frac{1}{3} I_{030} $                                                                                                  & $I_{3}$
  \\
  5  & 0          & 3              & \footnotesize $\frac{1}{6} I_{003} $                                                                                       & $I_{6}$
  \\

  6  & 0          & 3              & \footnotesize $-I_{012}$                                                                                                   & $I_{7}$
  \\
  7  & 0          & 4              & \footnotesize $\frac{1}{2} I_{004} $                                                                                       & $I_{5}$
  \\
  8  & 0          & 4              & \footnotesize $\frac{1}{2} I_{022} $                                                                                       & $I_{9}$
  \\
  9  & 0          & 5              & \footnotesize $\frac{1}{2} I_{014} $                                                                                       & $I_{8}$
  \\
  10  & 2          & 0              & \footnotesize $ I_{200} $                                                                                                  & $I_{10}$
  \\
  11 & 2          & 1              & \footnotesize $\frac{1}{2}I_{201}$                                                                                         & $I_{13}$
  \\
  12 & 2          & 1              & \footnotesize $ I_{210} $                                                                                                  & $I_{14}$
  \\
  13 & 2          & 2              & \footnotesize $\frac{1}{6}I_{200}I_{002}-I_{202}^{a} $                                                                    & $I_{15}$
  \\
  14 & 2          & 2              & \footnotesize $\frac{1}{2} I_{202}^{b} $                                                                                   & $I_{16}$
  \\
  15 & 2          & 2              & \footnotesize $-I_{211}$                                                                                                   & $I_{17}$
  \\
  16 & 2          & 2              & \footnotesize $-\frac{1}{6}I_{020}I_{200}+I_{220}$                                                                       & $I_{18}$
  \\
  17 & 2          & 3              & \footnotesize $ \frac{1}{4}I_{002}I_{201}-I_{203}$                                                                         & $I_{19}$
  \\
  18 & 2          & 3              & \footnotesize $\frac{1}{6}I_{210}I_{002}-\frac{1}{3}I_{012}I_{200}-I_{212}^{a} $                                           & $I_{20}$
  \\
  19 & 2          & 3              & \footnotesize $-I_{212}^{b} $                                                                                              & $I_{21}$
  \\
  20 & 2          & 3              & \footnotesize $\frac{1}{2} I_{221} $                                                                                       & $I_{22}$
  \\
  21 & 2          & 4              & \footnotesize $\frac{1}{2}I_{204}+\frac{1}{6}I_{200}I_{004} $                                                              & $I_{23}$
  \\
  22 & 2          & 4              & \footnotesize $-\frac{1}{6}I_{002}I_{211}-\frac{1}{4}I_{012}I_{201}-\frac{1}{2}I_{213}$                                    & $I_{24}$
  \\
  23 & $2$        & $4$            & \footnotesize 
  $\frac{1}{2}I_{222}$ & $I_{25}$\\
  24 & 4          & 0              & \footnotesize $ \frac{1}{2}I_{400}$                                                                                        & $I_{11}$
  \\
  25 & 4          & 1              & \footnotesize $ \frac{1}{2}I_{401}$                                                                                        & $I_{26}$
  \\
  26 & 4          & 1              & \footnotesize $ -I_{410}$                                                                                                  & $I_{27}$
  \\
  27 & 4          & 2              & \footnotesize $ \frac{1}{2}I_{402}$                                                                                        & $I_{28}$
  \\
  28 & 4          & 2              & \footnotesize $ -I_{411}$                                                                                                  & $I_{29}$
  \\
  29 & 6          & 0              & \footnotesize $ \frac{1}{6}I_{600}$                                                                                        & $I_{12}$
  \\
  30 & 6          & 1              & \footnotesize $ \frac{1}{9} I_{201}(I_{200})^{2}-\frac{1}{6}I_{401}I_{200}+\frac{1}{2}I_{601} $                            & $I_{30}$
  \\
  \bottomrule
\end{tabular}

  \caption{$\mm$-even Smith--Smith--Rivlin invariants (evaluated for $\bsigma'$ except line 1).}
  \label{tab:even-smith-invariants}
\end{table}

\begin{table}[ht]
  \centering
  \begin{tabular}{ccccc}
  \toprule
     & deg($\mm$) & deg($\bsigma$) & Evaluation                                                                                                      & Smith et al notation \\
  \midrule
  31 & 1          & 2              & \footnotesize $-J_{111}$                                                                                     & $L_1$            \\
  32 & 1          & 3              & \footnotesize $-J_{103}$                                                                                     & $L_2$            \\
  33 & 1          & 3              & \footnotesize $J_{112}$                                                                                      & $L_3$            \\
  34 & 1          & 3              & \footnotesize $J_{121}$                                                                                      & $L_4$            \\
  35 & 1          & 4              & \footnotesize $-J_{104}$                                                                                     & $L_5$            \\
  36 & 1          & 4              & \footnotesize $-2J_{113}^{a}+J_{113}^{b}-\frac{1}{6}I_{002}J_{111}$                                           & $L_6$            \\
  37 & 1          & 4              & \footnotesize $J_{122}$                                                                                      & $L_7$            \\
  38 & 1          & 4              & \footnotesize $ 3J_{113}^{a}-3J_{113}^{b}$                                                      & $J_1K_1$         \\
  39 & 1          & 4              & \footnotesize $ \frac{1}{2}I_{020}J_{111}+3J_{131}$                                                                     & $J_1K_2$         \\
  40 & 1          & 5              & \footnotesize $-\frac{1}{4}  J_{103}I_{002}-\frac{1}{2}J_{105}  $                                             & $L_8$            \\
  41 & 1          & 5              & \footnotesize $-\frac{1}{2}I_{002}J_{112}-\frac{1}{2}I_{003}J_{111}-3J_{114} $                               & $J_2K_1$         \\
  42 & 1          & 5              & \footnotesize $\frac{1}{2}I_{002}J_{121}+\frac{1}{2}I_{020}J_{103}+3J_{123} $                                & $J_1K_4$         \\
  43 & 1          & 5              & \footnotesize $-\frac{1}{2}I_{020}J_{112}-3J_{132} $                                                         & $J_2K_2$         \\
  44 & 3          & 1              & \footnotesize $J_{301}$                                                                                      & $L_9$            \\
  45 & 3          & 2              & \footnotesize $J_{302}$                                                                                      & $L_{10}$         \\
  46 & 3          & 2              & \footnotesize $ -J_{111}I_{200}+2J_{311}^{a}-J_{311}^{b} $                                                     & $L_{11}$         \\
  47 & 3          & 2              & \footnotesize $ 3J_{311}^{a}-I_{200}J_{111}-3J_{311}^{b}$                                                    & $J_1K_7$         \\
  48 & 3          & 3              &  \footnotesize $ -\frac{1}{2}\left( J_{103}I_{200}+J_{303}^{a}+J_{303}^{b}\right)$                             & $L_{12}$         \\
  49 & 3          & 3              &  \footnotesize $\frac{1}{2} \left( J_{103}I_{200}+J_{301}I_{002}+3(J_{303}^b-J_{303}^a)\right)$                                         & $J_1K_8$         \\
  50 & 3          & 3              & \footnotesize $-J_{312}^{a} $                                                                                & $J_0K_1$         \\
  51 & 3          & 3              & \footnotesize $J_{330} $                                                                                     & $J_0K_2$         \\

  52 & 3          & 3              & \footnotesize $I_{200}J_{112}-\frac{3}{2}\left( \frac{1}{2} I_{201}J_{111}+J_{312}^{a}+J_{312}^{b} \right) $ & $J_2K_7$         \\
  53 & 3          & 3              & \footnotesize $ \frac{1}{2}I_{020}J_{301}+3J_{321}$                                                          & $J_1K_{10}$      \\
  54 & 5          & 1              & \footnotesize $J_{501}$                                                                                      & $L_{13}$         \\
  55 & 5          & 1              & \footnotesize $-J_{510}$                                                                                     & $J_0K_7$         \\
  56 & 5          & 2              & \footnotesize $-J_{520} $                                                                                    & $J_0K_{10}$      \\
  57 & 5          & 2              & \footnotesize $J_{502}-\frac{1}{2}I_{201}J_{301} $                                                           & $J_0K_8$         \\
  58 & 5          & 2              & \footnotesize $-I_{200}J_{311}^a+\frac{1}{2}I_{400}J_{111}-2I_{210}J_{301}+3J_{511}$                         & $J_1K_{21}$      \\
  59 & 7          & 1              & \footnotesize $I_{400}J_{301}+I_{200}J_{501}-\frac{3}{2}J_{701} $                                            & $J_1K_6$         \\
  60 & 7          & 1              & \footnotesize $-J_{710} $                                                                                    & $J_0K_{21}$      \\
  61 & 9          & 0              & \footnotesize $-J_{900}$                                                                                     & $J_0K_6$         \\
  \bottomrule
\end{tabular}
  \caption{$\mm$-odd Smith--Smith--Rivlin invariants (evaluated for $\bsigma'$).}
  \label{tab:odd-smith-invariants}
\end{table}

\clearpage

\section{Application to cubic magneto-mechanical coupling at the domain scale}
\label{sec:domain-scale-coupling}

A first objective here is to obtain, in closed forms, general expressions of the  magneto-mechanical free enthalpy density for cubic symmetry such as those given in \autoref{sec:mechanical-modeling}, \emph{i.e.} as  polynomial expansions truncated at a given degree $\deg(\bsigma)$ in stress for instance. A second objective is to obtain the associated constitutive laws in an intrinsic manner (thanks to the tensorial expressions provided by invariants and to the use of the cubic projectors
$\bP^{d}$ and $\bP^{\overline d}$).

We have to point out that, in the physical problem studied, the sought enthalpies $\Psi^\star(\bsigma, \mm)$ are functions of the stress tensor $\bsigma$ and of the magnetization pseudo-vector $\mm$ (the action of the isometries $g$ on magnetization and stress being given by \eqref{eq:pseudo-vector-action}). Rather than the cubic subgroup $\octa$ of isometries (with $\mm$ a pseudo-vector), we have to consider the \emph{magnetic point group} $\octa^\varepsilon = \octa \times \set{\pm1}$ \cite{SC1984}. In practice, however, it boils down to the fact that $\Psi$ is $\octa$--invariant under the standard action~\eqref{eq:vector-action} (rather than~\eqref{eq:pseudo-vector-action})~\cite{WG2004}, implying that the sought enthalpies are even in magnetization $\mm$ \cite{Mau1990}. One can then use the integrity basis of $\RR[\RR^{3}\oplus \Sym^{2}(\RR^{3})]^{\octa}$ for full octahedral group $\octa$  in order to derive a polynomial form of cubic magneto-mechanical free energy density.

In order to shorten formulas,
\begin{itemize}
  \item the terms depending only on stress tensor $\bsigma$ are gathered in the pure elastic free energy density
        \begin{equation}\label{eq:psielas}
          \Psi^{\star e}=-\frac{1}{2}\bsigma:\bS:\bsigma
          ,\qquad
          \bepsilon^{e}=-\frac{\partial \Psi^{\star e}}{\partial \bsigma}=\bS: \bsigma,
        \end{equation}
        with $ \bepsilon^{e}$ the elastic strain tensor and where the compliance tensor $\bS$ satisfies $S_{ijkl}=S_{jikl}=S_{ijlk}=S_{klij}$,

  \item the terms depending only on the magnetization $\mm$ are gathered in the pure magnetic free energy density $\Psi^{\star\mu} (\mm)$ (classically of degree six in magnetization),

  \item the so-called first order magneto-mechanical terms, linear in stress and quadratic in $\mm$, are gathered in the free energy density
        $\Psi^{\star\mu \sigma}_{1} (\mm, \bsigma)$,

  \item the so-called second order magneto-mechanical terms, quadratic in stress and quadratic in $\mm$, are gathered in $\Psi^{\star\mu \sigma}_{2} (\mm, \bsigma)$ (defined in~\eqref{gibbslin} and~\eqref{Gibbsnonlin}),

  \item the dependency with respect to pair $(\mm, \bsigma)$ is written through the first cubic invariants $I_{\alpha\beta\gamma}$ of \autoref{tab:even-invariants} (and $I_{222}$ and
        $\tr \bsigma$), where the multi-degree in $(\mm, \bsigma^{d}, \bsigma^{\overline d})$  is $(\alpha, \beta, \gamma)$.
\end{itemize}
We have then
\begin{equation}\label{eq:quadratic-energy}
  \Psi^\star = \Psi^{\star e} (\bsigma)+\Psi_{1}^{\star\mu \sigma} (I_{\alpha\beta\gamma}, \tr \bsigma)+\Psi^{\star\mu \sigma}_{2} (I_{\alpha\beta\gamma},\tr \bsigma)+\Psi^{\star\mu} (\mm).
\end{equation}
The strain is given by the state law
\begin{equation*}
  \bepsilon = -\frac{\partial \Psi^\star }{\partial \bsigma},
\end{equation*}
derived, thanks to the chain rule, from
\begin{equation}\label{eq:dsigma}
  \frac{\partial f(I_{\alpha\beta\gamma}, \tr \bsigma) }{\partial \bsigma} = \frac{\partial f }{\partial I_{\alpha\beta\gamma}} \left(\bP^{d}:\frac{\partial I_{\alpha\beta\gamma}}{\partial \bsigma^{d}} + \bP^{\overline{d}}:\frac{\partial I_{\alpha\beta\gamma}}{\partial \bsigma^{d}}\right) + \frac{\partial f}{\partial \tr \bsigma} \id,
\end{equation}
using the fact, by \eqref{eq:PdPdbar0}--\eqref{eq:PdPdbar}, that the fourth-order cubic projectors are equal to
\begin{equation*}
  \frac{\partial \bsigma^{ d}}{\partial \bsigma} = \bP^{d},
  \qquad \frac{\partial \bsigma^{\overline{d}}}{\partial \bsigma} = \bP^{\overline{d}},
\end{equation*}
and finally that
\begin{equation*}
  \frac{\partial I_{\alpha\beta\gamma}}{\partial \tr \bsigma}=0,
  \quad
  \frac{\partial I_{\alpha\beta\gamma}}{\partial \bsigma^{d}}:\bP^{d}=\bP^{d}:\frac{\partial I_{\alpha\beta\gamma}}{\partial \bsigma^{d}}
  \quad
  \textrm{and}
  \quad
  \frac{\partial I_{\alpha\beta\gamma}}{\partial \bsigma^{d}}:\bP^{\overline{d}}=\bP^{\overline{d}}:\frac{\partial I_{\alpha\beta\gamma}}{\partial \bsigma^{d}}.
\end{equation*}
Magnetic field is given by the state law:
\begin{equation*}
  \mu_{0} \hh=-\frac{\partial \Psi^\star }{\partial \mm}=-\frac{\partial \Psi^\star }{\partial I_{\alpha\beta\gamma}}\frac{\partial I_{\alpha\beta\gamma}}{\partial \mm},
\end{equation*}
where $\mu_{0}$ is the vacuum permeability. The $\octa$-invariant $I_{\alpha\beta\gamma}$ are even in magnetization $\mm$. Moreover, their dependency in $\mm$ has been recast in \autoref{tab:even-invariants} as a dependency in
$\mm\otimes \mm$. This is also the case for $I_{200}=\norm{\mm}^2=\tr (\mm\otimes \mm)$ and $J_{111}^2$ preferably replaced by $I_{222}$ thanks to \eqref{J111square}.
Setting furthermore $(\mm\otimes \mm)^{d}=\bP^{d}:(\mm\otimes \mm)$ and $(\mm\otimes \mm)^{\overline{d}}=\bP^{\overline{d}}:(\mm\otimes \mm)$ we have then, for all the invariants $I_{\alpha\beta\gamma}$ of \autoref{tab:even-invariants},
\begin{align*}
  \frac{\partial I_{\alpha\beta\gamma}}{\partial \mm}
   & = \left(\frac{\partial I_{\alpha\beta\gamma}}{\partial (\mm\otimes \mm)^{d}}:\bP^{d}+\frac{\partial I_{\alpha\beta\gamma}}{\partial (\mm\otimes \mm)^{\overline{d}}}:\bP^{\overline{d}}\right):\frac{\partial (\mm\otimes \mm) }{\partial \mm}
  \\
   & = 2\left(\left(\frac{\partial I_{\alpha\beta\gamma}}{\partial (\mm\otimes \mm)^{d}}\right)^{\! \! d}+\left(\frac{\partial I_{\alpha\beta\gamma}}{\partial (\mm\otimes \mm)^{\overline{d}}}\right)^{\! \! \overline{d}\, }\right) \mm.
\end{align*}
and $\displaystyle \frac{\partial I_{200}}{\partial \mm}= 2 \mm$.

\subsection{Elastic free energy density}

Using~\autoref{tab:even-invariants}, the elastic free energy density writes:
\begin{equation}\label{eq:elastic-energy}
  \begin{aligned}
    \Psi^{\star e} (\bsigma) & = c_{020}\, I_{020}+c_{002}\, I_{002}+c_{010,010}\, (\tr \bsigma)^{2}                                                         \\
                             & = c_{020}\, \bsigma^{d}:\bsigma^{d} + c_{002}\, \bsigma^{\overline{d}}:\bsigma^{\overline{d}}+c_{010,010}\, (\tr\bsigma)^{2},
  \end{aligned}
\end{equation}
and the elastic strain as
\begin{equation}\label{eq:def:elas}
  \bepsilon^{e} = -\frac{\partial \Psi^{\star e} (\bsigma)}{\partial \bsigma}.
\end{equation}
This leads to the linear elastic constitutive relationship
\begin{equation*}
  \bepsilon^{e} = -2c_{020}\bsigma^{d} - 2c_{002}\bsigma^{\overline{d}}-2c_{010,010}\tr(\bsigma)\id ,
\end{equation*}
to be compared to the cubic linear Hooke's law~\cite{DM2011}
\begin{equation*}
  \bepsilon^{e} = \frac{1+\nu}{E} \bsigma^{d} + \frac{1}{2\mu} \bsigma^{\overline{d}} + \frac{1}{9\kappa}\tr\bsigma \, \id ,
\end{equation*}
where $E$ is the Young modulus, $\nu$ is the Poisson ratio, $\mu$ is the shear modulus and $\kappa=E/3(1-2\nu)$ is the compressibility modulus. We get thus
\begin{equation*}
  c_{020} = -(1+\nu)/2E, \qquad c_{002} = -1/4\mu, \qquad c_{010,010} = -1/18\kappa = -(1+\nu)/6E.
\end{equation*}

\begin{rem}
  An isotropic free energy density corresponds to the particular case where
  \begin{equation*}
    \mu=E/2(1+\nu), \quad \text{and} \quad \bsigma^{d}+\bsigma^{\overline{d}}=\bsigma'.
  \end{equation*}
  We get then
  \begin{equation*}
    \Psi^{\star e} = - \frac{1+\nu}{2E} \bsigma' : \bsigma' +\frac{1-2\nu}{6E}(\tr\bsigma)^{2}\, \id,
    \quad \text{and} \quad
    \bepsilon^{e}= \frac{1+\nu}{E} \bsigma' +\frac{1-2\nu}{3E}\tr\bsigma \, \id.
  \end{equation*}
\end{rem}

\subsection{Magnetic free energy density}

The purely magnetic part of \eqref{eq:quadratic-energy} which is the most general $\octa$-invariant polynomial of degree six in $\mm$, writes as
\begin{equation*}
  \Psi^{\star\mu} (\mm) = c_{200}I_{200} + c_{200,200}(I_{200})^{2} + c_{200,200,200}(I_{200})^{3}
  + c_{400}I_{400} + c_{200,400}I_{200}I_{400} + c_{600}I_{600},
\end{equation*}
where
\begin{equation*}
  I_{200}=\norm{\mm}^{2},
  \qquad
  I_{400}=\tr\big( (\mm \otimes \mm)^{\overline{d}\, 2} \big),
  \qquad
  I_{600}= \tr \big( (\mm \otimes \mm)^{\overline{d}\, 3} \big).
\end{equation*}
From the usual assumption $\norm{\mm}=m_{s}=\textit{constant}$ at the magnetic domain scale, some terms group together.
Introducing the direction cosines $\gamma_i$ of $\mm$ we further get:
\begin{equation*}
  I_{400} = 2 m_{s}^{4}\,(\gamma_{1}^{2}\gamma_{2}^{2}+\gamma_{1}^{2}\gamma_3^{2}+\gamma_{2}^{2}\gamma_3^{2}),
  \qquad
  I_{600} = 6 m_{s}^{6}\,\gamma_{1}^{2}\gamma_{2}^{2}\gamma_3^{2}.
\end{equation*}
Indeed, the standard cubic form~\cite{Boz1951,Cul1972}
\begin{equation*}
  \Psi^{\star\mu} (\mm)=K_{0}+K_{1}(\gamma_{1}^{2}\gamma_{2}^{2}+\gamma_{1}^{2}\gamma_3^{2}+\gamma_{2}^{2}\gamma_3^{2})+K_{2}(\gamma_{1}^{2}\gamma_{2}^{2}\gamma_3^{2})
\end{equation*}
is recovered by setting
\begin{align*}
  K_{0}= & c_{200} \, m_{s}^{2}+c_{200,200}\, m_{s}^{4}+c_{200,200,200}\,  m_{s}^{6},
  \\
  K_{1}= & 2m_{s}^{4} \left(c_{400}+c_{200,400}\, m_{s}^{2}\right),
  \\
  K_{2}= & 6m_{s}^{6} c_{600},
\end{align*}
where $K_{0}$, $K_{1}$ and $K_{2}$ are the so-called magneto-crystalline constants for cubic symmetry.

\begin{rem}
  An isotropic free energy density corresponds to the particular case where
  \begin{equation*}
    \Psi^{\star\mu} = a \norm{\mm}^{2} + b \norm{\mm}^{4} + c \norm{\mm}^{6} = a m_{s}^{2}+b m_{s}^{4}+c m_{s}^{6},
  \end{equation*}
  defining, here at degree six, no preferential direction.
\end{rem}

\subsection{First-order magneto-mechanical energy density term}

The so-called first-order magneto-mechanical term is linear in $\bsigma$ and quadratic in $\mm$. Its most general $\octa$-invariant polynomial expression is
\begin{equation*}
  \begin{aligned}
    \Psi^{\star\mu \sigma}_{1} (\mm,\bsigma) & = c_{210}I_{210}+c_{201}I_{201}+c_{200,010}I_{200}\tr \bsigma
    \\
                                             & = c_{210} (\mm \otimes \mm)^{d} : \bsigma^{d} +c_{201} (\mm \otimes \mm)^{\overline{d}} :\bsigma^{\overline{d}}  +c_{200,010}\norm{\mm}^{2}\tr\bsigma
  \end{aligned}
\end{equation*}
which can be recast as
\begin{equation*}
  \Psi^{\star\mu \sigma}_{1} (\mm,\bsigma)=\left(c_{210}(\mm \otimes \mm)^{d}+  c_{201} (\mm \otimes \mm)^{\overline{d}}+c_{200,010}\norm{\mm}^{2}\id\right):\bsigma .
\end{equation*}
The associated strain is
\begin{equation*}
  \bepsilon^{\mu}_{1}=-\frac{\partial \Psi^{\star\mu \sigma}_{1} (\mm,\bsigma)}{\partial \bsigma}=-c_{210}(\mm \otimes \mm)^{d}
  -c_{201}(\mm \otimes \mm)^{\overline{d}}-c_{200,010}\norm{\mm}^{2}\id .
\end{equation*}

This form is to be compared to the classical expression of the magnetostriction strain tensor reported in equation \eqref{eq:eps:mu} \cite{Cul1972,DTdLac1993}, illustrating the following correspondences between material constants:
\begin{equation*}
  \lambda_{100} = -\frac{2}{3}c_{210}m_{s}^{2} , \qquad
  \lambda_{111} = -\frac{2}{3}c_{201}m_{s}^{2} , \qquad
  \lambda_{v} = -3c_{200,010}m_{s}^{2}.
\end{equation*}

\begin{rem}
  Isotropy corresponds to the case $c_{201}=c_{210}$, $\lambda_{111}= \lambda_{100}$, \emph{i.e.} to
  \begin{equation*}
    \Psi^{\star\mu \sigma}_{1} (\mm,\bsigma) = -\frac{3 \lambda_{s}}{2 m_{s}^{2}} (\mm \otimes \mm)':\bsigma'-
    \frac{ \lambda_{v} }{3m_{s}^{2} } \norm{\mm}^{2}\, \tr\bsigma,
  \end{equation*}
  where $\lambda_s$ is the so-called isotropic magnetotrictive constant verifying $\lambda_s=\lambda_{100}=\lambda_{111}$.
\end{rem}

\subsection{Second-order magneto-mechanical free energy density term}

The so-called second-order magneto-mechanical term is quadratic both in $\bsigma$ and $\mm$. Its most general $\octa$-invariant polynomial expression is
\begin{equation}\label{eq:coupled-energy}
  \begin{aligned}
    \Psi^{\star\mu \sigma}_{2}(\mm,\bsigma) & = c_{220}I_{220}+c_{211}I_{211}+c_{202}^{a} I_{202}^{a} + c_{202}^{b}I_{202}^{b} + c_{210,010}I_{210}\tr \bsigma + c_{201,010}I_{201}\tr \bsigma
    \\
                                            & \quad + c_{200,020}I_{200}I_{020} + c_{200,002}I_{200}I_{002} + c_{200,010,010}I_{200}(\tr \bsigma)^{2},
  \end{aligned}
\end{equation}
which details as
\begin{align*}
  \Psi^{\star\mu \sigma}_{2} (\mm,\bsigma) & = c_{220} (\mm \otimes \mm)^{d}: \bsigma^{d\, 2} + c_{211} (\mm \otimes \mm)^{\overline{d}} : (\bsigma^{\overline{d}}\bsigma^{d}) + \left[c_{202}^{a} (\mm \otimes \mm)^{d} + c_{202}^{b} (\mm \otimes \mm)^{\overline{d}}\right]:\bsigma^{\overline{d}\, 2}
  \\
                                           & \quad + \tr \bsigma\, \left[c_{210,010} \, (\mm \otimes \mm)^{d} :\bsigma^{d}
  + c_{201,010} \, (\mm \otimes \mm)^{\overline{d}} :\bsigma^{\overline{d}} \right]                                                                                                                                                                                                                       \\
                                           & \quad + \norm{\mm}^{2}\left[ c_{200,020}\, \bsigma^{d}:\bsigma^{d}+c_{200,002} \, \bsigma^{\overline{d}}:\bsigma^{\overline{d}} + c_{200,010,010}\, (\tr \bsigma)^{2}\right].
\end{align*}
According to Mason~\cite{Mas1951}, the morphic effect involves six measurable material constants, \eqref{eq:coupled-energy} is in agreement with this assertion. Indeed, the last three terms in $\Psi^{\star\mu \sigma}_{2} (\mm,\bsigma)$,
\begin{equation*}
  m_{s}^{2} \left[c_{200,020}\,\bsigma^{d}:\bsigma^{d}+c_{200,002}\, \bsigma^{\overline{d}}:\bsigma^{\overline{d}}+ c_{200,010,010}\,
  (\tr \bsigma)^{2} \right]
\end{equation*}
do not depend on the orientation of the magnetization and thus cannot be distinguished from the purely elastic terms \eqref{eq:elastic-energy} in the total free energy density~\eqref{eq:quadratic-energy}. Therefore, the nine initial terms
are reduced to finally six measurable constants and we then set
\begin{equation*}
  c_{200,020} = c_{200,002} = c_{200,010,010}=0.
\end{equation*}
The associated strain
\begin{equation*}
  \bepsilon^{\mu}_{2} = - \frac{\partial \Psi^{\star\mu \sigma}_{2} (\mm,\bsigma)}{\partial \bsigma}
\end{equation*}
splits into three parts
\begin{equation*}
  \bepsilon^{\mu}_{2} = \bepsilon^{\mu\, d}_{2} + \bepsilon^{\mu\, \overline{d}}_{2} + \bepsilon^{\mu}_{2v},
\end{equation*}
where, using~\eqref{eq:PdPdbar0},
\begin{equation*}
  \bepsilon^{\mu\, d}_{2} = \bP^{d}:\bepsilon^{\mu}_{2},
  \qquad
  \bepsilon^{\mu\, \overline{d}}_{2} = \bP^{\overline{d}}:\bepsilon^{\mu}_{2},
  \qquad
  \bepsilon^{\mu}_{2v} = \frac{1}{3} \left(\tr \bepsilon^{\mu}_{2}\right) \id.
\end{equation*}
We get
\begin{align*}
  \bepsilon^{\mu \, d}_{2}            & = - 2c_{220}\,(\mm \otimes \mm)^{d} \bsigma^{d} - \frac{1}{2} c_{211}\,\left((\mm \otimes \mm)^{\overline{d}} \bsigma^{\overline{d}} + \bsigma^{\overline{d}}(\mm \otimes \mm)^{\overline{d}} \right)^{d} - c_{210,010}\, (\tr \bsigma)\, (\mm \otimes \mm)^{d},
  \\
  \bepsilon^{\mu \, \overline{d}}_{2} & = -\frac{1}{2} c_{211}\left((\mm \otimes \mm)^{\overline{d}} \bsigma^{d} + \bsigma^{d}(\mm \otimes \mm)^{\overline{d}} \right)^{\overline{d}} - c_{201,010} \, (\tr\bsigma)\, (\mm \otimes \mm)^{\overline{d}}
  \\
                                      & \quad - \left[\left(c_{202}^{a} (\mm \otimes \mm)^{d} + c_{202}^{b} (\mm \otimes \mm)^{\overline{d}}\right)\bsigma^{\overline{d}}
    + \bsigma^{\overline{d}}\left(c_{202}^{a}(\mm \otimes \mm)^{d} + c_{202}^{b} (\mm \otimes \mm)^{\overline{d}}\right)
    \right]^{\overline{d}},
  \\
  \bepsilon^{\mu}_{2v}                & = - \left(c_{210,010}\,(\mm \otimes \mm)^{d} :\bsigma^{d} + c_{201,010} \, (\mm \otimes \mm)^{\overline{d}} :\bsigma^{\overline{d}}\right) \id.
\end{align*}
In the canonical cubic basis $(\ee_{i})$, they correspond respectively to the deviatoric diagonal part, the out-of-diagonal part and the volumetric (hydrostatic) part of $\bepsilon^{\mu}_{2}$.

\begin{rem}
  Using example~\ref{ex:R3-plus-S2-SO3}, the most general expression for an isotropic quadratic free energy density writes
  \begin{equation*}
    \Psi^{\star\mu\sigma}_{2} = A\,(\mm \otimes \mm)': \bsigma^{\prime \, 2} + B(\tr \bsigma)(\mm \otimes \mm)': \bsigma' + \norm{\mm}^{2}\left(C\tr(\bsigma^{\prime\, 2}) + D \,(\tr \bsigma)^{2}\right).
  \end{equation*}
  At the magnetic domain scale, since we have
  \begin{equation*}
    \mm \cdot (\bsigma^{\prime \, n }\mm) = (\mm \otimes \mm): \bsigma^{\prime  \,n}, \quad \text{and} \quad (\tr(\mm \otimes \mm)\, \id): \bsigma^{\prime  n} = \norm{\mm}^{2} \tr(\bsigma^{\prime  n}),
  \end{equation*}
  for $n=1,2$, where $\bsigma^{\prime \, n} = (\bsigma^\prime)^ {n}$, we get
  \begin{equation*}
    \Psi^{\star\mu\sigma}_{2} = A\, (\mm \otimes \mm)': \bsigma^{\prime \, 2} + B\,(\tr \bsigma)(\mm \otimes \mm)': \bsigma' + m_{s}^{2} \left(C \tr( \bsigma^{\prime\, 2})+ D \,(\tr \bsigma)^{2}\right).
  \end{equation*}
  There are indeed four material constants, in accordance with Kraus~\cite{Kra1988}, but the last two terms associated with material constants $C$ and $D$ do not depend on the orientation of the magnetization and cannot be distinguished from the purely elastic terms. Therefore, we set
  \begin{equation*}
    C = D = 0.
  \end{equation*}
  Since
  \begin{equation*}
    \bsigma' = \bsigma^{d} + \bsigma^{\overline{d}} = \bP^{d}:\bsigma+\bP^{\overline{d}}:\bsigma, \quad \text{and} \quad \bsigma^{d}:\bsigma^{\overline{d}} = 0,
  \end{equation*}
  an isotropic free energy density $\Psi^{\star\mu\sigma}_{2}$ expands as
  \begin{multline*}
    \Psi^{\star\mu\sigma}_{2} = A \left((\mm \otimes \mm)^{d} + (\mm \otimes \mm)^{\overline{d}}\right) : (\bsigma^{d\,2}+2\bsigma^{d}\bsigma^{\overline{d}} + \bsigma^{\overline{d}\, 2})\, + \, B (\tr \bsigma)\left((\mm \otimes \mm)^{d}: \bsigma^{d} + (\mm \otimes \mm)^{\overline{d}}: \bsigma^{\overline{d}} \right).
  \end{multline*}
  By comparison with \eqref{eq:coupled-energy} and using the fact that $(\mm \otimes \mm)^{\overline{d}}: \bsigma^{d\, n}=0$ and
  \begin{equation*}
    \left. (\mm \otimes \mm)^{d}: (\bsigma^{d}\bsigma^{\overline{d}})\right. = \tr[(\mm \otimes \mm)^{d}\bsigma^{d}\bsigma^{\overline{d}}] = \left. ((\mm \otimes \mm)^{d}\bsigma^{d}):\bsigma^{\overline{d}}\right. = 0,
  \end{equation*}
  one sees that isotropy corresponds to
  \begin{equation*}
    A = c_{220} = c_{202}^{a} = c_{202}^{b} = \frac{1}{2} c_{211}, \quad \text{and} \quad B = c_{210,001} = c_{210,010}.
  \end{equation*}
  Isotropic second-order magnetostriction strain tensor expresses then as
  \begin{equation*}
    \bepsilon_{2}^{\mu} = - A \, ( (\mm \otimes \mm)' \bsigma^{\prime}+\bsigma^{\prime}(\mm \otimes \mm)')'-B\,\left((\tr \bsigma)(\mm \otimes \mm)'+ ((\mm \otimes \mm)': \bsigma')\, \id\right)
  \end{equation*}
  where $A$ and $B$ can be expressed as functions of the Kraus constants~\cite{Kra1988} $\lambda^{\prime}_s$ and $\lambda^{\prime\prime}_s$
  \begin{equation*}
    A = \frac{3}{2} \frac{\lambda^{\prime\prime}_s}{m_s^2} , \qquad
    B = -\frac{1}{2} \frac{3\lambda^\prime_s+\lambda^{\prime\prime}_s}{m_s^2}.
  \end{equation*}
\end{rem}

\subsection{Higher-order free energy densities}

In the above subsections, we have recovered known free energy densities using several bi-homogeneous polynomial expansions and expressed them in an intrinsic form. The knowledge of a minimal integrity basis of 29 invariants given in~\autoref{tab:even-invariants} together with $I_{222}$ (by theorem \ref{thm:main-octa} and remark~\ref{rem:I222}) allows us to easily generate the most general expression of a bi-homogeneous invariant polynomial of an arbitrary bi-degree $(\alpha, \beta)$ in $(\mm, \bsigma)$. The methodology used avoids the tedious calculations associated with the enforcement of the cubic symmetry for constitutive tensors of order six or more (see for instance \cite{Smi1994}). Any higher order, both in magnetization and in stress, are now reachable for cubic symmetry and the number of invariants involved at each degree in magnetization and stress is given in~\autoref{tab:coeffs} (which includes the invariants function of $I_{200}=\norm{\mm}^{2}$).

Note however, that the numbers of bi-homogeneous invariants involved in the free energy density expressed at the magnetic domain scale are not necessarily equal to the numbers of associated material parameters. Indeed, the further relationship $\norm{\mm}^{2}=m_{s}^{2}$ (a material parameter then) due to constant norm of the magnetization $\mm$ has to be considered. We propose to illustrate this subtlety in the case of an energy density of degree 2 in magnetization and 3 in stress (bi-degree $(2,3)$). According to \autoref{tab:coeffs}, 20 bi-homogeneous invariants span the vector space $\RR[V^{\prime}]^{\octa}_{2,3}$:
\begin{itemize}
  \item $I_{203}$, $I_{002}I_{201}$, $I_{003}I_{200}$
  \item $I_{212}^{a}$, $I_{212}^{b}$, $I_{002}I_{210}$, $I_{012}I_{200}$, $I_{202}^{a}\tr \bsigma$, $I_{202}^{b}\tr \bsigma$, $I_{200}I_{002}\tr \bsigma$
  \item $I_{221}$, $I_{020}I_{201}$, $I_{201}(\tr \bsigma)^{2}$, $I_{211}\tr \bsigma$
  \item $I_{020}I_{210}$, $I_{030}I_{200}$, $I_{200}(\tr \bsigma)^{3}$, $I_{210}(\tr \bsigma)^{2}$, $I_{220}\tr \bsigma$, $I_{020}I_{200}\tr \bsigma$
\end{itemize}
Among these invariants, 6 depends on $I_{200}= \norm{\mm}^{2}=m_{s}^{2}=\textit{constant}$. Being related to the non-linear
elasticity terms $I_{003}$, $I_{012}$, $I_{002}\tr \bsigma$, $I_{030}$, $I_{020}\tr \bsigma$, which are magnetization independent, they cannot be distinguished from these initial elasticity terms. Therefore, 14 coefficients instead of 20 are necessary to describe cubic magneto-elastic phenomena at the domain scale, of degree 2 in magnetization and 3 in stress.

This illustrating case can be generalized to any order in a straightforward manner as only the invariant $I_{200}$ is a constant. Consequently, for a degree $\alpha$ in magnetization, the invariants of a bi-degree $(\alpha, \beta)$ are all those of degree $\alpha$ in magnetization not depending on $I_{200}$ plus all those of degree $\alpha-2$ in $\mm$ multiplied by $I_{200}$. Thus, it is enough to subtract two consecutive lines from~\autoref{tab:coeffs} to obtain the number of material coefficients at each bi-degree. These results are summarized in \autoref{tab:coeffs:mat}.

\begin{table}[ht]
  \centering
  \begin{tabular}{c|ccccccccccc}
    \toprule
    \diagbox{$\mm$}{$\bsigma$} & 0  & 1 & 2  & 3  & 4   & 5   & 6   & 7   & 8    & 9    & 10
    \\
    \midrule
    0                          & -- & 1 & 3  & 6  & 11  & 18  & 32  & 48  & 75   & 111  & 160
    \\
    2                          & 0  & 2 & 6  & 14 & 31  & 60  & 106 & 180 & 288  & 442  & 659
    \\
    4                          & 1  & 3 & 10 & 24 & 53  & 102 & 185 & 312 & 504  & 777  & 1161
    \\
    6                          & 1  & 4 & 13 & 34 & 73  & 144 & 262 & 444 & 717  & 1112 & 1660
    \\
    8                          & 1  & 5 & 17 & 42 & 95  & 186 & 378 & 576 & 933  & 1443 & 2162
    \\
    10                         & 1  & 6 & 20 & 52 & 115 & 228 & 375 & 708 & 1146 & 1748 & 2661
    \\
    \bottomrule
  \end{tabular}
  \caption{Number of material parameters associated with different bi-degrees in $(\mm, \bsigma)$.}
  \label{tab:coeffs:mat}
\end{table}

\section{Conclusion}
\label{sec:conclusion}

This paper takes as a study support the magnetization and stress couple $(\mm, \bsigma)$ in order to build a free energy density, expressed as a polynomial function of $\mm$ and $\bsigma$, suitable for magnetostrictive materials with cubic microstructures.
An application of so-called second-order\footnote{in fact of degree 2 in stress and magnetization.} magneto-elasticity was shown to account for the non-monotonic sensitivity to stress of both the magnetic susceptibility and the magnetostriction of some soft magnetic materials~\cite{Hub2019}. It introduces a sixth-order constitutive tensor, the morphic tensor. The extension to higher bi-degrees in $(\mm, \bsigma)$ involves then constitutive tensors of high order, and using tensors of order higher than 6 (with cubic symmetry) becomes rather tedious using such techniques.

A polynomial formulation has thus been preferred rather than constitutive tensor formulations in order to express the free energy density for any bi-degree in magnetization and stress. For this purpose, new minimal integrity bases
$\set{\tr \bsigma, I_{\alpha \beta \gamma}, J_{\alpha \beta \gamma}}$ for the orientation preserving octahedral group ($\octa^+$) and $\set{\tr \bsigma, I_{\alpha \beta \gamma}}$ for the full octahedral group ($\octa$) have been computed. The $\octa^+$--integrity basis is constituted of 60 invariants (\emph{i.e.} one less than in the initial Smith-Smith-Rivlin integrity basis) when the $\octa$--integrity basis is constituted of 30 invariants (all even in $\mm$). Furthermore, we have proved that both the proposed $\octa^+$-- and $\octa$--integrity bases are minimal.

The novelty is that these new integrity bases are expressed in a simple intrinsic way. Contrary to Smith-Smith-Rivlin invariants, the proposed cubic invariants are not expressed in a particular coordinate system (and thus, nor their partial derivatives with respect to $\bsigma$ and to $\mm$, as provided in a systematic manner in~\autoref{sec:domain-scale-coupling}). Theses new (coordinate free) invariants allow general expressions of the free energy density $\Psi^\star(\bsigma, \mm)$ and the magneto-elastic coupling at the magnetic domains scale, for material exhibiting cubic symmetry.

This work leads moreover to the following perspectives concerning magneto-elasticity:
\begin{enumerate}
  \item \textit{the introduction in the free energy density of terms of higher degree in stress and/or functions of invariants}: this introduction in a multiscale model may help to model the saturation of the magnetoelastic phenomena (magnetostriction especially) still not reachable yet \cite{Hub2019};

  \item \textit{the extension to macroscopic constitutive laws}: this extension is possible by considering the magnetization $\mm$ as a non-constant norm quantity. Of course the number of invariants increases compared to the case $\norm{\mm}=m_s=\textit{constant}$ and cubic symmetry might be not relevant (the texturized materials are usually orthotropic). The cubic symmetry remains however relevant for highly texturized materials and/or materials that can be described by their texture components \cite{DANIEL2020}. It is worth pointing out that any cubic energy density is also a function (not necessarily polynomial) $\Psi^\star(\tr \bsigma, I_{\alpha\beta\gamma})$, because an integrity basis is also a functional basis;

  \item \textit{the identification of material constants via dedicated experiments:} this is another relevant perspective. We may include in this strategy some complementary hypothesis like incompressibility. Experimental results (magnetostriction and magnetic behavior under stress) can on the other hand help to define the relevant invariants.
\end{enumerate}

Last, the methodology developed is not limited to magneto-elasticity and may be used for the modeling of any other coupled phenomena
involving, in the cubic symmetry case, a second order tensor and a (pseudo-)vector.

\appendix

\section{Proofs}
\label{sec:proofs}

In this section, we provide proofs of theorems~\ref{thm:main-octa-plus} and~\ref{thm:main-octa}.

\begin{proof}[Proof of theorem~\ref{thm:main-octa-plus}]
  Observe first that, if $\set{J_{1}, \dotsc , J_{r}}$ is a minimal integrity basis of the invariant algebra
  \begin{equation*}
    \RR[\RR^{3} \oplus \HH^{2}(\RR^{3})]^{\octa^{+}},
  \end{equation*}
  then, $\set{\tr \bsigma, J_{1}, \dotsc , J_{r}}$ is a minimal integrity basis of
  \begin{equation*}
    \RR[\RR^{3} \oplus \Sym^{2}(\RR^{3})]^{\octa^{+}}.
  \end{equation*}
  Now, verifying that the set $\mathcal{F}$ of invariants obtained by those in~\autoref{tab:even-invariants} (but omitting the first invariant $\tr \bsigma$) and those of~\autoref{tab:odd-invariants} is a minimal generating set for $\RR[\RR^{3} \oplus \HH^{2}(\RR^{3})]^{\octa^{+}}$ has been achieved using the algorithm described in \autoref{sec:algo} and the Computer Algebra System \emph{Macaulay2}~\cite{Macau2}, a software specialized in algebraic geometry and effective in polynomial computations. The crucial degree bound $N = 12$, used to terminate the algorithm was obtained \emph{a priori} by theorem~\ref{thm:octa-plus-bound} in~\autoref{sec:degree-bounds}.
\end{proof}

\begin{rem}\label{rem:proof}
  In the proof, the algorithm of \autoref{sec:algo} has been applied to the family $\mathcal{F}$ of the 60 invariants in~\autoref{tab:even-invariants} and~\autoref{tab:odd-invariants}. Since the output was $\mathcal{F}$ itself, we conclude that $\mathcal{F}$ is minimal basis of $\RR[\RR^{3} \oplus \HH^{2}(\RR^{3})]^{\octa^{+}}$. But, when applied to the family of the 61 invariants in Smith--Smith--Rivlin paper~\cite[Section 6]{SSR1963}, the algorithm lead to the discovery of the superfluous invariant $I_{25}$.
\end{rem}

We will now proceed with the proof of theorem~\ref{thm:main-octa}. For this goal, we will introduce the \emph{Reynolds projector}, a very efficient tool in invariant theory. This operator $R_{G}$ is defined for any finite group $G$ and any representation $V$ of $G$.

\begin{defn}
  The Reynolds projector is a linear projector from $\RR[V]$ onto the invariant algebra $\RR[V]^{G}$, defined as
  \begin{equation}\label{eq:Reynolds-projector}
    R_{G}(p) := \frac{1}{\abs{G}} \sum_{g \in G} g \star p, \qquad p \in \RR[V],
  \end{equation}
  where $\abs{G}$ is the order of the group $G$ and $g\star p:=p(\rho(g)^{-1}\vv)$ for $\vv\in V$.
\end{defn}

\begin{proof}[Proof of theorem \ref{thm:main-octa}]
  The full cubic group $\octa$ can be decomposed as the disjoint union
  \begin{equation}\label{eq:octa-decomposition}
    \octa = \octa^{+} \sqcup I_{c}\octa^{+},
  \end{equation}
  where $\octa^{+}$ is the subgroup of positive symmetries of the cube and where $I_{c}$ is the central symmetry.
  We get thus in particular, by~\eqref{eq:Reynolds-projector}--\eqref{eq:octa-decomposition},
  \begin{equation*}
    R_{\octa}(p) = \frac{1}{2}\left(R_{\octa^{+}}(p) + I_{c}\star R_{\octa^{+}}(p)\right).
  \end{equation*}
  Consider now a polynomial $p$ invariant under $\octa$. We have then $R_{\octa}(p)=p$ and since $p$ is obviously invariant under $\octa^{+}$, we get thus
  \begin{equation}\label{eq:Reynnolds-decomposition}
    p = R_{\octa}(p) = \frac{1}{2}\left(p + I_{c}\star p\right).
  \end{equation}
  Now $p$ is invariant by $\octa^{+}$, we can thus write $p$ (using theorem~\ref{thm:main-octa-plus}) as a polynomial function of the 29 invariants $I_{\alpha\beta\gamma}$ from~\autoref{tab:even-invariants} and the 31 invariants $J_{\alpha'\beta'\gamma'}$ from~\autoref{tab:odd-invariants} and will write
  \begin{equation*}
    p = P(I_{\alpha\beta\gamma}, J_{\alpha'\beta'\gamma'}).
  \end{equation*}
  where, we have moreover
  \begin{equation*}
    I_{c} \star I_{\alpha\beta\gamma} = I_{\alpha\beta\gamma}, \qquad I_{c} \star J_{\alpha'\beta'\gamma'} = -J_{\alpha'\beta'\gamma'}.
  \end{equation*}
  We get thus
  \begin{equation*}
    I_{c} \star p = P(I_{\alpha\beta\gamma}, -J_{\alpha'\beta'\gamma'}).
  \end{equation*}
  Now, using~\eqref{eq:Reynnolds-decomposition}, we get
  \begin{equation*}
    p = \frac{1}{2} \left(P(I_{\alpha\beta\gamma},J_{\alpha'\beta'\gamma'}) + P(I_{\alpha\beta\gamma},-J_{\alpha'\beta'\gamma'})\right).
  \end{equation*}
  Therefore, in the expansion of $p$, only monomials in $I_{\alpha\beta\gamma}$ and $J_{\alpha'\beta'\gamma'}$ with an even number of $J_{\alpha'\beta'\gamma'}$ remain. We conclude that $\RR[V]^\octa$ is generated by the $I_{\alpha\beta\gamma}$ and the products $J_{\alpha\beta\gamma} J_{\alpha''\beta''\gamma''}$
  which are invariant under $\octa$. This family is thus a generating set but it may not be \emph{minimal}. However, the fact that it is a generating set provides us with the following precious information: the maximum degree of the generators of a minimal basis is less than
  \begin{equation*}
    \max \set{ \max \deg(I_{\alpha\beta\gamma}), \max \deg(J_{\alpha'\beta'\gamma'} J_{\alpha''\beta''\gamma''})} \le 18,
  \end{equation*}
  because the maximal degree is obtained for the square of $J_{900}$ (the last invariant in~\autoref{tab:odd-invariants}). The family provided in the theorem has been checked to be a minimal integrity basis using the algorithm described in this section with degree bound $N=18$, which is known \emph{a priori} by the argument above.
\end{proof}

\section{An algorithm to extract a minimal integrity basis}
\label{sec:algo}

In this section, we formulate an algorithm which, starting from a finite set $\mathcal{F}=\set{I_{1}, \dotsc , I_{p}}$ of homogeneous invariants, checks if this set generates the invariant algebra, and, if the answer is positive, extracts from it a minimal integrity basis $\mathcal{MB} = \set{I_{i_{1}}, \dotsc , I_{i_{r}}}$. It is derived from an algorithm introduced in~\cite{DOADK2020} to clean up a generating set $\mathcal{F}$ from its redundant elements and thus produce a minimal integrity basis $\mathcal{MB}$. Since in the present case, $\mathcal{F}$ is not assumed to be a generating set, the procedure in~\cite{DOADK2020} has to be modified.

In these algorithms, calculations are executed in finite dimensional vector spaces of polynomials of a given degree and a degree bound $N$ must be furnished in order to terminate the computations. If we know that the family $\mathcal{F}$ is generating, then this bound $N$ can be defined as the maximum degree of the polynomials in the family $\mathcal{F}$ (as done in \cite{OKA2017,DOADK2020}). However, in the present case, we do not know \textit{a priori} that the family $\mathcal{F}$ is a generating set and more mathematics are required. Since we deal with a finite group, a degree bound on the generators of the invariant algebra is already known. It is given by Noether's theorem (see~\autoref{sec:degree-bounds}) and equals the order of the group. In our case, this order is $48$ for $\octa$ and $24$ for $\octa^{+}$, much higher than our computation means, since \emph{the computation time is exponential in the total degree $n$} as illustrated in~\autoref{fig:temps}. We had therefore to use more sophisticated tools of group theory in order to reduce this bound. This has been done in~\autoref{sec:degree-bounds}, where we have reduced this \textit{a priori} bound to $N=12$ for $\octa^{+}$ and $N=18$ for the full octahedral group $\octa$. These lower bounds have allowed to decrease drastically the computation time, approximately from one month to one second !

\begin{figure}[ht]
  \centering
  \includegraphics[scale=0.8]{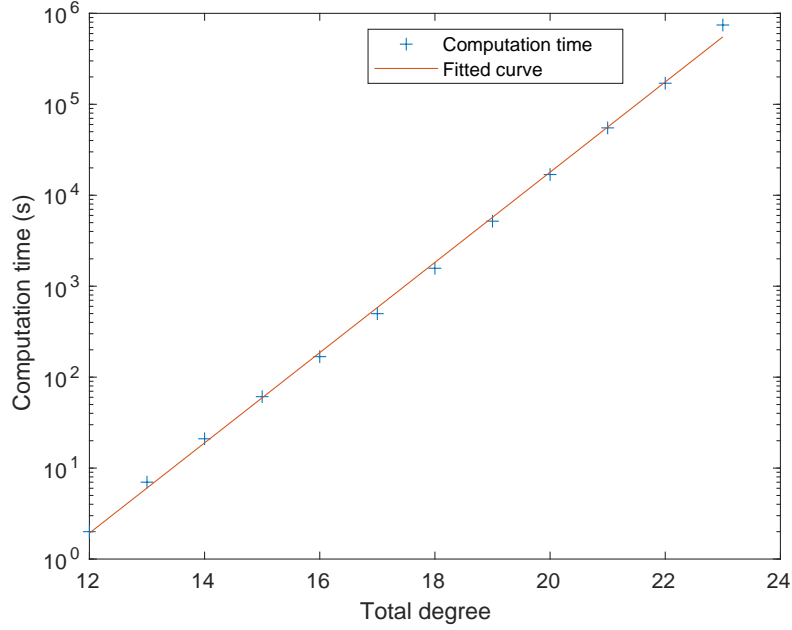}
  \caption{Computation time versus total degree $n$ on an eight-core machine with a CPU frequency of 1197 MHz and 32 GB of RAM.}
  \label{fig:temps}
\end{figure}

In the following, we set $V'=\RR^{3} \oplus \HH^{2}(\RR^{3})$ and $G$ stands either for $\octa$ or $\octa^{+}$. We recall that $V'$ splits in three irreducible representations of $G$, which have been described by the three independent variables $\mm$, $\bsigma^{d}$ and $\bsigma^{\overline{d}}$. Moreover, the invariant algebra $\RR[V']^{G}$ can be split into the direct sum of vector spaces
\begin{equation*}
  \RR[V']^{G} = \bigoplus_{n=0}^{\infty} \RR[V']^{G}_{n},
\end{equation*}
of homogeneous polynomials in $\mm$, $\bsigma^{d}$ and $\bsigma^{\overline{d}}$ which are of total degree $n$. We need thus to check that each finite dimensional vector space $\RR[V']^{G}_{n}$ is spanned by the products $I_{1}^{\mathsf{a}_{1}} \dotsm I_{p}^{\mathsf{a}_{p}}$ of total degree $n$. This procedure seems, at first, never-ending, because we need to check this for all $n \in \NN$. However, if we can show \textit{a priori} that degrees of the homogeneous polynomials in a minimal integrity basis are smaller than some bound $N$ (see~\autoref{sec:degree-bounds} for explicit estimations of such a bound), then, one needs only to check this for $n \le N$. Indeed, then, all homogeneous polynomials of degree $n> N$ are reducible.

The inputs/output of the algorithm are

\begin{itemize}
  \item \textbf{Inputs} :
        \begin{itemize}
          \item[(a)] A finite set $\mathcal{F}$ of homogeneous polynomial invariants in $\RR[V']^{G}$.
          \item[(b)] The coefficients $a_{n}$ (for $n \le N$) of the Hilbert series
                \begin{equation*}
                  H_{\rho}(z) = \sum a_{n} \, {z}^{n}, \quad \text{where} \quad a_{n} = \dim \left(\RR[V']^{G}_{n}\right).
                \end{equation*}
        \end{itemize}
  \item \textbf{Output} : A minimal integrity basis $\mathcal{MB}$ extracted from $\mathcal{F}$, if $\mathcal{F}$ is a generating set, and otherwise an error message.
\end{itemize}

\begin{rem}
  For the proper octahedral group $G=\octa^{+}$ and $V'=\RR^{3} \oplus \HH^{2}(\RR^{3})$, the Hilbert series writes
  \begin{align*}
    H_{(V',\octa^+)}(z) & = \frac{1}{24}\left( \frac{1}{(1-z)^8} + \frac{6}{(1-z^2)^4} + \frac{3}{(1-z^2)^4} + \frac{8(1-z)}{(1-z^3)^3} + \frac{6}{(1-z^4)^2}\right)
    \\
                        & = 1+3\,{z}^{2} + 6\,{z}^{3} + 17\,{z}^{4} + 33\,{z}^{5} + 81\,{z}^{6} + 141\,{z}^{7}+282\,{z}^{8}
    \\
                        & \quad + 480\,{z}^{9} + 828\,{z}^{10} + 1326\,{z}^{11} + 2137\,{z}^{12}+O \left( {z}^{13} \right).
  \end{align*}
  For the full octahedral group $G=\octa$ and $V'=\RR^{3} \oplus \HH^{2}(\RR^{3})$, it writes
  \begin{align*}
    H_{(V',\octa)}(z) & = \frac{1}{24}\left( \frac{1+3z^2}{(1+z)^3(1-z)^8} + \frac{6}{(1-z)(1-z^2)^4} + \frac{3}{(1-z)(1-z^2)^4} \right.
    \\
                     & \quad + \left. \frac{8(1-z)}{(1+z^3)(1-z^3)^3} + \frac{6}{(1+z)(1-z^4)^2}\right)
    \\
                     & = 1+3\,{z}^{2}+5\,{z}^{3} + 13\,{z}^{4}+22\,{z}^{5} + 52\,{z}^{6} + 84\,{z}^{7} + 164\,{z}^{8}
    \\
                     & \quad + 268\,{z}^{9} + 456\,{z}^{10} + 714\,{z}^{11} + 1141\,{z}^{12} + 1697\,{z}^{13} + 2560\,{z}^{14}
    \\
                     & \quad + 3692\,{z}^{15} + 5310\,{z}^{16} + 7413\,{z}^{17} + 10317\,{z}^{18} + O\left({z}^{19}\right).
  \end{align*}
\end{rem}

Let $n_{0}$ be the smallest positive integer $n$ for which the coefficient $a_{n}$ of the Hilbert series $H_{\rho}$ does not vanish ($a_{0}$ is always equal to $1$), and for each $n \ge 1$, define $\mathcal{F}_{n}$ as the subset of elements of $\mathcal{F}$ of degree $n$. The algorithm consists in three main steps, and is summarized below.

\subsection*{The algorithm}

\begin{description}
  \item[Initialization $n = n_{0}$] Extract a subfamily $\mathcal{B}_{n_{0}}$ from $\mathcal{F}_{n_{0}}$ of maximal rank (and linearly independent) in the vector space $\RR[V']^{G}_{n_{0}}$. If $\rank(\mathcal{B}_{n_{0}}) = a_{n_{0}}$, increment $n$ by one, otherwise, return an error message.

  \item[Iteration step $n$ ($n_{0} < n \le N$)]
        Suppose that we have obtained, at step $n-1$, the family $\mathcal{B}_{n-1}$. Note that $\mathcal{B}_{n-1}$ may contain homogeneous polynomials of different degrees but all of them are lower than $n$.
        \begin{itemize}
          \item Generate the family $\mathcal{R}_{n}$ of all \emph{reducible homogeneous invariants} of degree $n$ which can be written as products of polynomials in $\mathcal{B}_{n-1}$ (this can be done by solving a Diophantine equation~\cite[Section 6]{DOADK2020}).
          \item If $\rank(\mathcal{R}_{n}) = a_{n}$, then set $\mathcal{B}_{n} := \mathcal{B}_{n-1}$ and increment $n$ by one.
          \item Otherwise, extract from $\mathcal{F}_{n}$ a minimal subset $\mathcal{I}_{n}$ such that $\rank(\mathcal{R}_{n}\cup \mathcal{I}_{n}) = a_{n}$. This can be done by ordering the elements of $\mathcal{F}_{n}$ (if $\mathcal{F}_{n}$ is not empty), adding iteratively an element of $\mathcal{F}_{n}$ to $\mathcal{B}_{n-1}$ and checking the rank of the new family.
          \item If, at the end, the new set $\mathcal{B}_{n}$ satisfies $\rank(\mathcal{B}_{n}) = a_{n}$, increment $n$ by one, otherwise, return an error message.
        \end{itemize}

  \item [Termination $n=N$] If all the steps of the algorithm have matched, then, one can conclude that the output $\mathcal{MB} := \mathcal{B}_{N}$ is a minimal integrity basis of $\RR[V']^{G}$.
\end{description}

\begin{rem}
  In practice, it may be advantageous to refine the algorithm in the following way, in order to optimize the computation time. Rather than using the mono-graduation given by the total degree, we can use the tri-graduation defined by the decomposition
  \begin{equation*}
    \RR[V']^{G} = \bigoplus_{\alpha, \beta, \gamma} \RR[V']^{G}_{\alpha \beta \gamma},
  \end{equation*}
  of multi-homogeneous polynomials of respective degree $\alpha$, $\beta$, $\gamma$ in $\mm$, $\bsigma^{d}$ and $\bsigma^{\overline{d}}$. We need then to define a total order on multi-indices $(\alpha,\beta,\gamma)$. The appropriate choice corresponds to the \emph{graded lexicographic order}, which is denoted by $\preceq$ and defined as follows:
  \begin{equation*}
    (\alpha_{1},\beta_{1},\gamma_{1}) \preceq (\alpha_{2},\beta_{2},\gamma_{2}),
  \end{equation*}
  if the total degree $\alpha_{1}+\beta_{1}+\gamma_{1}$ of $(\alpha_{1},\beta_{1},\gamma_{1})$ is lower than the total degree $\alpha_{2}+\beta_{2}+\gamma_{2}$ of $(\alpha_{2},\beta_{2},\gamma_{2})$, or if they have the same total degree and the first non-vanishing difference $\alpha_{1}-\alpha_{2}$, $\beta_{1}-\beta_{2}$, $\gamma_{1}-\gamma_{2}$ is negative. For instance, we have
  \begin{equation*}
    (0,0,1) \preceq (0,1,0) \preceq (1,0,0) \preceq (0,0,2) \preceq (0,1,1) \preceq \dotsb
  \end{equation*}
  Then, the corresponding algorithm is the same with the only difference that the iteration must be done using multi-degrees $(\alpha, \beta, \gamma)$ and the graded lexicographic order, rather than the total degree $n = \alpha + \beta + \gamma$ and its natural order. In that case, one must use the tri-graded Hilbert series $H_{\rho}(z_{m},z_{d},z_{\bar d})$ (see~\eqref{eq:Hilbert-octa-plus-tri-graded} and~\eqref{eq:Hilbert-octa-tri-graded}).
\end{rem}

\section{Degree bounds}
\label{sec:degree-bounds}

The termination of the algorithm provided in~\autoref{sec:algo} is crucially dependent on the existence of a bound on the total degree of the generators of a minimal integrity basis which \emph{must be known a priori}. Group theory is required to estimate such a bound. In this appendix, we will state a few fundamental results on such bounds for the problem we consider \textit{i.e.} the action of the cubic groups $\octa$ and $\octa^{+}$ on some vector space $V$.

Let $G$ be a finite group \emph{acting linearly} on a vector space $V$. We call such an action a \emph{representation} of $G$ on $V$. It is known since 1916, thanks to a theorem of Emmy Noether~\cite{Noether1916}~\cite[Theorem 2.1.4]{Stu1993}, that the invariant algebra $\RR[V]^{G}$ can always be generated by homogeneous polynomials with total degree lower than $\abs{G}$, the order of the group. In particular, the algebra $\RR[V]^{G}$ is finitely generated and the degrees of the elements of a minimal integrity basis are lower than $\abs{G}$. However, this result is far from optimal and there are many situations where this bound can be lowered~\cite{Sch1991a,NS2002}.

As already stated, a minimal integrity basis of $\RR[V]^{G}$ is not unique but the degrees of the generators of a minimal integrity basis and their number are well defined and independent of the choice of a particular basis (like the number of elements in a basis of a vector space is independent of a particular basis and defines the dimension of the space). We will denote by $\beta(G,V)$ this maximum degree. The problem is that there is no general algorithm to compute this number $\beta(G,V)$, but to compute a minimal integrity basis of $\RR[V]^{G}$. Let us now define
\begin{equation*}
  \beta(G) = \sup_{V} \beta(G,V),
\end{equation*}
where $V$ runs over all finite dimensional representations of $G$. Of course, by Noether's theorem~\cite{Noether1916} we get
\begin{equation}
  \label{eq:noether}
  \beta(G) \le \abs{G},
\end{equation}
but there are finite groups for which a better bound can be explicitly computed. This is the case, for instance, for the dihedral group $D_{n}$ of index $n$. This group can be realized as the subgroup of isometries of the plane generated by the rotation by angle $2\pi/n$ and by the reflection with respect to the $x$-axis. It is of order $2n$ and Noether's bound leads to
\begin{equation*}
  \beta(D_{n})\le 2n .
\end{equation*}
However, it was proved in~\cite{Sch1991a} the following optimal result, which is an important improvement compared to Noether's bound.

\begin{thm}\label{thm:dihedral-group}
  Let $D_{n}$ be the dihedral group of index $n$. Then,
  \begin{equation*}
    \beta(D_{n}) = n + 1 .
  \end{equation*}
\end{thm}

The second ingredient required to achieve our goal is an abstract but very useful result from group theory, lemma~\ref{lem:beta-norm} below, which was formulated and proved in~\cite[Lemma 3.1]{Sch1991a}. To understand its statement, recall first that a subgroup $N$ of $G$ is said to be \emph{normal} if it stable by conjugacy, which means that
\begin{equation*}
  gNg^{-1} = N, \qquad \forall g \in G.
\end{equation*}
The important property is that if $N$ is a normal subgroup of $G$, then the \emph{quotient space} (the set of left classes)
\begin{equation*}
  G/N := \set{gN;\; g \in G}
\end{equation*}
is also a group, where the group operation is defined as
\begin{equation*}
  g_{1}N \star g_{2}N := (g_{1}g_{2})N.
\end{equation*}

\begin{lem}\label{lem:beta-norm}
  Let $G$ be a finite group and $N$ be a normal subgroup of $G$. Then,
  \begin{equation*}
    \beta(G) \le \beta(G/N) \beta(N).
  \end{equation*}
\end{lem}

We are now able to formulate the main result of this section and provide a proof for it.

\begin{thm}\label{thm:octa-plus-bound}
  Let $\octa^{+}$ be the subgroup of positive isometries which preserve the cube. Then, we have
  \begin{equation*}
    \beta(\octa^{+}) \le 12 .
  \end{equation*}
\end{thm}

\begin{rem}
  It has not been checked that the bound proposed in theorem~\ref{thm:octa-plus-bound} is optimal.
\end{rem}

In order to make our proof as simple as possible, we will recall first the following facts.

\begin{enumerate}
  \item The group $\octa^{+}$ is isomorphic (as an abstract group) to the permutation group of four elements, noted $\mathfrak{S}_{4}$. It corresponds, indeed, to the permutations of the four main diagonals of the cube (see~\cite[15.4 p. 273]{coxeter1961} and \cite{FH1991}, for instance).
  \item The subset of the permutation group $\mathfrak{S}_{4}$ generated by double transpositions
        \begin{equation*}
          \set{e,(12)(34), (13)(24), (14)(23)}
        \end{equation*}
        is a normal subgroup of $\mathfrak{S}_{4}$ which is isomorphic (as an abstract group) to the dihedral group $D_{2}$ and the quotient group $\mathfrak{S}_{4}/D_{2}$ is isomorphic (as an abstract group) to $\mathfrak{S}_{3}$, the permutation group of $3$ elements.
  \item The permutation group of $3$ elements $\mathfrak{S}_{3}$ is isomorphic (as an abstract group) to the dihedral group $D_{3}$~\cite{coxeter1961,FH1991}.
\end{enumerate}

\begin{proof}[Proof of theorem~\ref{thm:octa-plus-bound}]
  Since $\octa^{+}$ is isomorphic (as an abstract group) to $\mathfrak{S}_{4}$, we have
  \begin{equation*}
    \beta(\octa^{+}) = \beta(\mathfrak{S}_{4}).
  \end{equation*}
  Now, using the fact that the subgroup of $\mathfrak{S}_{4}$ generated by double transpositions is a normal subgroup isomorphic (as an abstract group) to  $D_{2}$ and that the quotient group $\mathfrak{S}_{4}/D_{2}$ is isomorphic (as an abstract group) to $\mathfrak{S}_{3}$ we deduce from lemma~\ref{lem:beta-norm} that
  \begin{equation}\label{eq:S4-S3}
    \beta(\mathfrak{S}_{4}) \le \beta(\mathfrak{S}_{3})\beta(D_{2}).
  \end{equation}
  Finally, since $\mathfrak{S}_{3}$ is isomorphic (as an abstract group) to $D_{3}$, we get that $\beta(\mathfrak{S}_{3}) = \beta(D_{3})$ and hence that
  \begin{equation*}
    \beta(\mathfrak{S}_{4}) \le \beta(D_{3})\beta(D_{2}) = 12
  \end{equation*}
  by virtue of theorem~\ref{thm:dihedral-group}. This achieves the proof.
\end{proof}


\end{document}